\documentclass[a4paper,10pt,reqno]{amsart}

\usepackage[english]{babel}
\usepackage[utf8]{inputenc}
\usepackage[T1]{fontenc}
\usepackage{lmodern}
\usepackage{emptypage}
\usepackage{microtype}
\usepackage{indentfirst}
\usepackage{multicol}
\usepackage{comment}
\usepackage{todonotes}
\usepackage[margin=2.5cm]{geometry}

\usepackage{graphicx}
\usepackage{subfig}
\usepackage{tikz}
\usepackage{tikz-cd}
\usepackage{circuitikz}
\usetikzlibrary{}
\usepackage{caption}
\usepackage{float}

\usepackage{amsmath}
\usepackage{amssymb}
\usepackage{amsthm}
\usepackage{amsfonts}
\usepackage{mathrsfs}
\usepackage{mathtools}
\usepackage{physics}
\usepackage{nicematrix}
\usepackage{enumerate}
\usepackage{enumitem}

\usepackage{enumitem}

\newcommand{\upmodels}{\mathrel{\rotatebox[origin=c]{90}{\scalebox{0.7}{$\models$}}}}

\usepackage{tabularx}
\usepackage{longtable} 
\definecolor{ForestGreen}{RGB}{34,139,34}
\usepackage{colortbl}
\usepackage{nicematrix}

\usepackage[linesnumbered,ruled,vlined]{algorithm2e}
\SetKwInput{KwInput}{Input}                
\SetKwInput{KwOutput}{Output}              

\usepackage{csquotes}

\usepackage{hyperref}
\hypersetup{
    colorlinks=true,
    linkcolor=black,
    filecolor=black,      
    urlcolor=black,
    citecolor=black,
    pdftitle={Overleaf Example},
    pdfpagemode=FullScreen,
    }

\newtheorem{theorem}{Theorem}[section]

\newtheorem{lemma}[theorem]{Lemma}

\newtheorem{definition}[theorem]{Definition}
\newtheorem{proposition}[theorem]{Proposition}

\newtheorem{example}[theorem]{Example}

\newtheorem{remark}[theorem]{Remark}

\newcommand{\minus}{\scalebox{0.75}[1.0]{$-$}}

\title{Decoding Trombetti-Zhou codes: a new
syndrome-based decoding approach}
\author{Chunlei Li, Angelica Piccirillo, Olga Polverino and Ferdinando Zullo}
\date{\today}

\begin{document}

\begin{abstract}
     In 2019, Trombetti and Zhou introduced a new family of $\mathbb{F}_{q^n}$-linear Maximum Rank Distance (MRD) codes over $\mathbb{F}_{q^{2n}}$. For such codes we propose a new syndrome-based decoding algorithm. It is well known that a syndrome-based decoding approach relies heavily on a parity-check matrix of a linear code. Nonetheless, Trombetti-Zhou codes are not linear over the entire field $\mathbb{F}_{q^{2n}}$, but only over its subfield $\mathbb{F}_{q^{n}}$. Due to this lack of linearity, we introduce the notions of $\mathbb{F}_{q^{n}}$-generator matrix and $\mathbb{F}_{q^{n}}$-parity-check matrix for a generic $\mathbb{F}_{q^{n}}$-linear rank-metric code over $\mathbb{F}_{q^{rn}}$ in analogy with the roles that generator and parity-check matrices play in the context of linear codes. Accordingly, we present an $\mathbb{F}_{q^n}$-generator matrix and $\mathbb{F}_{q^n}$-parity-check matrix for Trombetti-Zhou codes as evaluation codes over an $\mathbb{F}_q$-basis of $\mathbb{F}_{q^{2n}}$. This relies on the choice of a particular basis called \emph{trace almost dual basis}. Subsequently, denoting by $d$ the minimum distance of the code, we show that if the rank weight $t$ of the error vector is strictly smaller than $\frac{d-1}{2}$, the syndrome-based decoding of Trombetti-Zhou codes can be converted to the decoding of Gabidulin codes of dimension one larger. On the other hand, when $t=\frac{d-1}{2}$, we reduce the decoding to determining the rank of a certain matrix. The complexity of the proposed decoding for Trombetti-Zhou codes is also discussed.
\end{abstract}

\maketitle

{\textbf{Keywords}: }Rank-metric codes, Trombetti-Zhou codes, syndrome-based decoding algorithm.

\section{Introduction}
    As coding theory plays a significant role in data transmission and cryptography, rank-metric codes have become a promising alternative to codes in the Hamming metric \cite{RMCandtheirapplications,RankCodesGab}.
    \par 
    Rank-metric codes are a set of $m\times\ell$ matrices over a finite field $\mathbb{F}_q$ endowed with the rank distance, i.e.,
    \[
    d(A,B)=\mathrm{rank}(A-B)
    \]
    for each $A,B\in\mathbb{F}_q^{m\times\ell}$. Equivalently, they can be regarded as a set of vectors of length $\ell$ over the extension field $\mathbb{F}_{q^m}$, or, in the case of square matrices, as a set of linearized polynomials. \par The genesis of rank-metric codes can be traced back to Delsarte's seminal paper \cite{delsarte1978} in 1978, where they were first defined as $q$-analogs of the usual linear error-correcting codes endowed with the Hamming distance. Delsarte also introduced a \textit{Singleton-like bound}, which specifies the constraint that each code's parameters must adhere to. Rank-metric codes that meet this bound are referred to as \textit{Maximum Rank Distance} (MRD) codes, in analogy with Maximum Distance Separable (MDS) codes in the Hamming metric. In the same paper, from the perspective of bilinear forms, Delsarte constructed the first family of linear MRD codes. In 1985 \cite{gabidulin_TheoryOfCodesWithMaximumRankDistance}, Gabidulin independently presented the same class of MRD codes, though defined as the evaluation of linearized polynomials. In fact, these codes can be considered as $q$-analogs of Reed-Solomon codes. Despite being first introduced by Delsarte, they are usually called \textit{Gabidulin codes}. Additionally, Gabidulin exhibited several properties of codes in rank-metric and proposed an efficient decoding algorithm for Gabidulin codes. In 1991 \cite{gabidulinCodes_ApplicationsInDistributedStorage2}, Roth independently rediscovered such a family of codes in an attempt to correct crisscross errors, as the Hamming metric is unsuitable for these error patterns.
    \par 
    In recent years, the exponential growth of global IP traffic has led to an increased focus on the efficiency and reliability of digital network communications. The reliability of information delivery on wireless channels is often in contrast with a low delay transmission. Furthermore, the usage of retransmissions, resending of packets which have been either damaged or lost, is undesired in conjunction with high-rate transmission. A significant advancement in this field was made by Ahlswede, Cai, Li and Yeung \cite{Network_information_flow} who proved that Random Linear Network Coding (RLNC) is an effective forward error correction method for transmitting data. This approach has the potential to enhance network capacity, reliability and efficiency, however it is highly susceptible to errors. A breakthrough in this direction was made by Silva, K\"otter and Kschischang \cite{RankMetric_approach_to_ErrorControl_in_RNetworkCoding} who introduced a rank-metric approach to error control in random network coding. They show how codes in the rank-metric can be ``lifted'' to subspace codes in such a way that the rank distance between two codewords is reflected into the subspace distance between their lifted images. In particular, when lifted rank-metric codes are employed, the decoding problem for random network coding can be rephrased in rank metric terms. This let the authors to a generalized decoding problem for rank-metric codes that embraces two dual parameters, erasures and deviations, that correspond to partial information about the error matrix. Finally, they also proposed a decoding algorithm for Gabidulin codes that takes into account erasures and deviations.
    \par
    Although the interest in rank-metric codes has increased exponentially due to their application to network coding, many other applications have also emerged, including code-based cryptography \cite{RMCandtheirapplications,gabidulinCodes_ApplicationsInCryptography}, distributed storage \cite{RMCandtheirapplications,gabidulinCodes_ApplicationsInDistributedStorage2}, space-time coding \cite{gabidulinCodes_ApplicationsInSpace-TimeCoding} and finite geometry \cite{PhDThesis_Santonastaso}, to name a few. Because of their several applications, the development of decoding algorithms for new classes of error-correcting codes in the rank metric represents a critical research topic. 
    \par As previously outlined, the most well-known family of MRD codes is the one of Gabidulin codes. They were later generalized to the so-called \emph{generalized Gabidulin codes} by Kshevetskiy and Gabidulin \cite{GeneralizedGabCodes} and to the so-called \emph{generalized twisted Gabidulin codes} by Sheekey \cite{sheekey_twistedGC}. Subsequently, Otal and \"{O}zbudak introduced a new family of additive MRD codes known as \emph{additive generalized twisted Gabidulin (AGTG) codes} which contains all the previous classes of codes as subfamilies. Other families of MRD codes that are not equivalent to any subfamily of AGTG codes include the non-additive MRD codes introduced by Otal and \"{O}zbudak \cite{SomeNewnon-additiveMRDcodes}, the MRD codes designed by Sheeky \cite{sheekeyNewSemifieldsAndNewMRDCodesFromSkewPolynomialRings}, \emph{Trombetti-Zhou} (TZ) codes \cite{TZ-codes_anewfamily} and the new MRD codes by Lobillo, Santonastaso and Sheeky arising from skew polynomial rings \cite{lobillo2025quotients}.
    \par In particular, the construction proposed by Trombetti and Zhou in \cite{TZ-codes_anewfamily} builds upon a highly structured set of linearized polynomials in $\mathbb{F}_{q^{2n}}$ whose first and last coefficients are taken independently from the subfield of linearity $\mathbb{F}_{q^{n}}$ and which strictly relies on an element $\gamma\in\mathbb{F}_{q^{2n}}$ whose norm is a non square element in $\mathbb{F}_q$. These polynomials give rise to $\mathbb{F}_{q^n}$-linear MRD codes over $\mathbb{F}_{q^{2n}}$. In \cite{On_interpolation-based_decoding_of_a_class_of_maximum_rank_distance_codes} the authors presented an interpolation-based decoding algorithm to decode them. They make use of the properties of the Dickson matrix associated with a linearized polynomial of given rank, together with the modified Berlekamp-Massey algorithm. When the rank of the error vector reaches the unique decoding radius, the problem reduces to solving a quadratic polynomial, rather than a projective polynomial, which guarantees that the proposed decoding algorithm runs in polynomial-time. 
    \subsection{Our contribution}
    By extending the notion of generator matrix for linear codes to what we call $\mathbb{F}_{q^n}$-generator matrix for $\mathbb{F}_{q^n}$-linear codes over $\mathbb{F}_{q^{rn}}$, in analogy with e.g. \cite{BallGamboaLavrauw2023}, we further investigate the notion of parity-check matrix to that of $\mathbb{F}_{q^n}$-parity-check matrix. We subsequently focus on Trombetti-Zhou codes, which can naturally be viewed as linearized polynomial evaluation codes over an $\mathbb{F}_q$-basis of $\mathbb{F}_{q^{2n}}$. An $\mathbb{F}_{q^n}$-generator matrix for a Trombetti-Zhou code is a $2k\times 2n$ matrix with entries in $\mathbb{F}_{q^{2n}}$ whose rows consist of one of its $\mathbb{F}_{q^n}$-basis. The introduction of an $\mathbb{F}_{q^n}$-parity-check matrix for Trombetti-Zhou codes relies on the choice of a particular basis called \emph{trace almost dual basis}. We propose a bounded minimum distance syndrome-based decoding algorithm for Trombetti-Zhou codes with minimum rank distance $d$. When the rank weight of the error vector is strictly smaller than $\frac{d-1}{2}$, the decoding of Trombetti-Zhou codes can be converted to the decoding of Gabidulin codes of dimension one larger. Nevertheless, the decoding of Gabidulin codes cannot be pushed to the limit case in which the rank weight of the error vector reaches the unique decoding radius, i.e., when $d-1$ is even and $t=\frac{d-1}{2}$, as more information is needed. Our algorithm overcomes this limitation and is capable of correcting error vectors of rank weight at most $\left\lfloor\frac{d-1}{2}\right\rfloor$ and, in particular, when $d-1$ is even, error vectors of rank weight $\frac{d-1}{2}$ with entries in $\mathbb{F}_{q^n}$. Although the interpolation-based decoding algorithm for Trombetti-Zhou codes proposed in \cite{On_interpolation-based_decoding_of_a_class_of_maximum_rank_distance_codes} achieves a polynomial-time complexity as low as $\mathcal{O}(n^2)$ over $\mathbb{F}_{q^{2n}}$, our syndrome-based approach lends itself more naturally to the design of an efficient decoder that can properly exploit erasures and deviations in line with the algorithm proposed by Silva, K\"otter and Kschischang in \cite{RankMetric_approach_to_ErrorControl_in_RNetworkCoding} for Gabidulin codes in connection with RLNC.
    \subsection{Outline of the paper} The structure of the paper is as follows. Section \ref{sec: Preliminaries} provides an overview of rank-metric codes in both vector and linearized polynomial form, outlines their link and recalls several tools that will be employed throughout the paper. In Section \ref{sec: Overview and new insights into Trombetti-Zhou codes} we describe Trombetti-Zhou codes by recalling known properties and presenting new structural insights that exploit their $\mathbb{F}_{q^n}$-linearity. We also generalize these new structures to $\mathbb{F}_{q^n}$-linear rank-metric codes in $\mathbb{F}_{q^{rn}}$. Section \ref{sec: Decoding of Trombetti-Zhou codes: a new syndrome-based decoding approach} presents the theoretical foundations of the syndrome-based decoding process for Trombetti–Zhou codes, along with the final algorithm (Algorithm \ref{alg:Decoding of Trombetti-Zhou codes: a new syndrome-based decoding approach.}). In Section \ref{sec: Conclusions and open problems} we conclude with future research directions arising from the results of the paper. Finally, the Appendix includes some technical proofs. 
    \section{Preliminaries}
    \label{sec: Preliminaries}
    We begin by fixing the notation that will be used throughout the paper. Let $q$ be a prime power and let $\mathbb{F}_q$ denote the finite field with $q$ elements and $\mathbb{F}_{q^m}$ an extension field of $\mathbb{F}_q$ of degree $m$.  
    \par
    Let $r,n,m$ be positive integers such that $m=rn$ and let $\mathbb{F}_{q^n}$ be a subfield of $\mathbb{F}_{q^m}$. Let $x_0,x_1,\ldots,x_{h}\in\mathbb{F}_{q^m}$, then $\left<x_0,x_1,\ldots,x_{h}\right>_{\mathbb{F}_{q^n}}$ denotes the $\mathbb{F}_{q^n}$-linear subspace of $\mathbb{F}_{q^m}$ spanned by $x_0,x_1,\ldots,x_{h}$. Moreover we denote by $\mathrm{dim}_{q^n}\left(\left<x_0,x_1,\ldots,x_{h}\right>_{\mathbb{F}_{q^n}}\right)$ its $\mathbb{F}_{q^n}$-dimension.  
    \par We denote by upper case letters matrices $A\in\mathbb{F}_{q^m}^{s\times t}$, by $A^{\top}\in\mathbb{F}_{q^m}^{t\times s}$ the transpose of a matrix $A$ and by $0$ the null matrix. If $A\in\mathbb{F}_{q^m}^{s\times t}$, we will frequently denote by $\mathrm{Rowspan}_{\mathbb{F}_{q^n}}(A)\subseteq\mathbb{F}_{q^m}^t$ the $\mathbb{F}_{q^n}$-linear subspace spanned by the rows of $A$ and by $\mathrm{Colspan}_{\mathbb{F}_{q^n}}(A)\subseteq\mathbb{F}_{q^m}^s$ the $\mathbb{F}_{q^n}$-linear subspace spanned by the columns of $A$. 
    \medskip
    \\
    As already outlined in the introduction, rank-metric codes arise in a matrix framework, however they can be viewed as a set of vectors of length $\ell$ over the extension field $\mathbb{F}_{q^m}$ or, in the case of square matrices, as a set of linearized polynomials.
    \medskip
    \\
    Vector representation is currently the most widely used in applications. The \textit{rank} \textit{support} of a vector $\underline{x}\in\mathbb{F}_{q^m}^{\ell}$ is the $\mathbb{F}_q$-linear space of $\mathbb{F}_{q^m}$ spanned by its entries:
    \[
        \mathrm{supp}(\underline{x})=\left<x_0,x_1,\dots,x_{\ell-1}\right>_{\mathbb{F}_q}\subseteq\mathbb{F}_{q^m}.
    \]
    The \textit{rank} \textit{weight} $w(\underline{x})$ of a vector $\underline{x}\in\mathbb{F}_{q^m}^{\ell}$ is the dimension of the linear space spanned over $\mathbb{F}_q$ by its entries, i.e., the $\mathbb{F}_q$-dimension of its rank support:
    \[
        w(\underline{x})=\mathrm{dim}_q\left(\left<x_0,x_1,\dots,x_{\ell-1}\right>_{\mathbb{F}_q}\right).
    \] 
    The set of vectors in $\mathbb{F}_{q^m}^{\ell}$ can be endowed with the \textit{rank distance} defined by
    \[
        d(\underline{x},\underline{y})=w(\underline{x}-\underline{y})
    \]
    for all $\underline{x},\underline{y}\in\mathbb{F}_{q^m}^{\ell}$. 
    \begin{definition}
        A non-empty subset $\mathcal{C}\subseteq\mathbb{F}_{q^m}^{\ell}$ endowed with the rank distance defined above is called (\textit{vector}) \textit{rank-metric code}.
    \end{definition}
    \begin{definition}
        Let $r,n,m$ be positive integers such that $m=rn$ and let $\mathbb{F}_{q^n}$ be a subfield of $\mathbb{F}_{q^m}$. Let $\ell$ be a positive integer. A rank-metric code $\mathcal{C}\subseteq\mathbb{F}_{q^m}^{\ell}$ is an $\mathbb{F}_{q^n}$\textit{-linear rank-metric code} if for all $\underline{x},\underline{y}\in\mathcal{C}$ and $\alpha,\beta\in\mathbb{F}_{q^n}$, we have $\alpha\underline{x}+\beta\underline{y}\in\mathcal{C}$.
    \end{definition}
    Henceforth, the term \emph{linear code} will be used to refer to those codes that are linear over $\mathbb{F}_{q^m}$. Otherwise, the subfield on which the code is linear will be specified. The fundamental parameters of an $\mathbb{F}_{q^n}$-linear rank-metric codes are defined as follows. 
    \begin{definition}
        Let $r,n,m$ be positive integers such that $m=rn$ and let $\mathbb{F}_{q^n}$ be a subfield of $\mathbb{F}_{q^m}$. Let $\ell$ be a positive integer and let $\mathcal{C}\subseteq\mathbb{F}_{q^m}^{\ell}$ be an $\mathbb{F}_{q^n}$-linear rank-metric code. Its elements are called \textit{codewords}. The integer $\ell$ is the \textit{length} of the code. The \textit{dimension} of $\mathcal{C}$ is the dimension as $\mathbb{F}_{q^n}$-vector space and, if $\mathcal{C}$ is the non-zero code, the \textit{minimum} \textit{rank} \textit{distance} is, as usual,  
        \[
            d\coloneqq d(\mathcal{C})=\min\{d(\underline{x},\underline{y})\colon \underline{x},\underline{y}\in\mathcal{C},~\underline{x}\neq\underline{y}\}.
        \]
        We will refer to the $\mathbb{F}_{q^n}$-linear rank-metric code $\mathcal{C}\subseteq\mathbb{F}_{q^m}^\ell$ as a $q^n\minus[\ell,\kappa,d]_{q^m/q}$ code if it has length $\ell$, $\mathbb{F}_{q^n}$-dimension $\kappa\leq r\ell$ and minimum rank distance $d$. In the event that the minimum rank distance is either unknown or irrelevant, we write $q^n\minus[\ell,\kappa]_{q^m/q}$.
    \end{definition}
    If $\mathcal{C}$ is a linear code, we will simply refer to it as an $[\ell,\kappa,d]_{q^m/q}$ code. 
    \begin{definition}
        Let $\mathcal{C}$ be a linear $[\ell,\kappa]_{q^m/q}$ rank-metric code. $G\in\mathbb{F}_{q^m}^{\kappa\times \ell}$ is a \emph{generator matrix} of the code if its rows generate $\mathcal{C}$ as an $\mathbb{F}_{q^m}$-linear space while $H\in\mathbb{F}_{q^m}^{(\ell-\kappa)\times \ell}$ is a \emph{parity-check matrix} if it has the code as kernel, i.e., $\ker\left(H^{\top}\right)=\mathcal{C}$. 
    \end{definition}
    Depending on the employed non-degenerate bilinear form (see \cite[p. 52]{taylor1992geometry}), several orthogonal complements can be defined. The most widely used in coding theory is the Delsarte dual code.
    \begin{definition}
        The Delsarte dual code of an $[\ell,\kappa]_{q^m/q}$ code $\mathcal{C}\subseteq\mathbb{F}_{q^m}^{\ell}$ is the $[\ell,\ell-\kappa]_{q^m/q}$ code defined as
        \[
            \mathcal{C}^{\perp}\coloneqq\left\{\underline{x}\in\mathbb{F}_{q^m}^{\ell}\colon\left<\underline{x},\underline{c}\right>=\displaystyle\sum_{i=0}^{\ell-1}x_ic_i=0\text{ for all }\underline{c}\in\mathcal{C}\right\}\subseteq\mathbb{F}_{q^m}^{\ell}.
        \]
        Moreover, if $H$ is the parity-check matrix of $\mathcal{C}$, then $\mathcal{C}^{\perp}=\mathrm{Rowspan}_{q^m}(H)$.
    \end{definition}
    \noindent For completeness, we have recalled the definition of Delsarte duality, however, this work mainly relies on a different notion of orthogonal space based on the trace function that will be introduced in the following section.
    \medskip
    \\
    \noindent A natural link exists between the matrix and vector frameworks. In particular, a linear rank-metric code can be associated with a code over $\mathbb{F}_q^{m\times\ell}$ with the same cardinality and metric properties. Let $\underline{\beta}=(\beta_0,\beta_1,\dots,\beta_{m-1})$ denote an ordered $\mathbb{F}_q$-basis of $\mathbb{F}_{q^m}$; the map $\mathrm{ext}_{\underline{\beta}}$
    \begin{align*}
        \mathrm{ext}_{\underline{\beta}}\colon\mathbb{F}_{q^m}^\ell&\to\mathbb{F}_{q}^{m\times\ell}\\
            \underline{x}&\mapsto X=\begin{pNiceMatrix}
                    x_{0,0} & x_{0,1} & \dots & x_{0,\ell-1} \\
                    x_{1,0} & x_{1,1} & \dots & x_{1,\ell-1} \\
                    \vdots & \vdots & \ddots & \vdots \\
                    x_{m-1,0} & x_{m-1,1} & \dots & x_{m-1,\ell-1} \\
                \end{pNiceMatrix}
    \end{align*}
    where $x_i=x_{0,i}\beta_0+x_{1,i}\beta_1+\dots+x_{m-1,i}\beta_{m-1}$, is an $\mathbb{F}_q$-linear isometry between $(\mathbb{F}_{q^m}^\ell,d)$ and $(\mathbb{F}_{q}^{m\times \ell},d)$. For all $\underline{x}\in\mathbb{F}_{q^m}^\ell$, $\mathrm{ext}_{\underline{\beta}}(\underline{x})$ is called the matrix expansion (by columns) of the vector $\underline{x}$ with respect to the $\mathbb{F}_q$-basis $\underline{\beta}$ of $\mathbb{F}_{q^m}$. Given an $\mathbb{F}_q$-basis $\underline{\beta}$ of $\mathbb{F}_{q^m}$, the $\underline{\beta}$-column support (resp. $\underline{\beta}$-row support) of a vector $\underline{x}\in\mathbb{F}_{q^m}^\ell$ is the $\mathbb{F}_q$-vector space generated by the columns (resp. rows) of $\mathrm{ext}_{\underline{\beta}}(\underline{x})$. It follows from standard linear algebra arguments that the $\underline{\beta}$-column support (resp. $\underline{\beta}$-row support) of a vector $\underline{x}\in\mathbb{F}_{q^m}^\ell$ does not depend on the choice of the basis $\underline{\beta}$. This directly leads to the following definition. 
    \begin{definition}
        Let $\underline{x}\in\mathbb{F}_{q^m}^\ell$. The rank weight of $\underline{x}$ is either the dimension of the $\mathbb{F}_q$-linear subspace of $\mathbb{F}_{q^m}$ generated by its entries or, equivalently, the dimension of the $\mathbb{F}_q$-linear subspace generated by the columns (or the rows) of $\mathrm{ext}_{\underline{\beta}}(\underline{x})$, where $\underline{\beta}$ is any $\mathbb{F}_q$-basis of $\mathbb{F}_{q^m}$.
    \end{definition}
    In analogy to codes in the Hamming metric, the parameters of a rank-metric code must obey a Singleton-like bound on their cardinality.
    \begin{theorem}
        \label{thm: Singleton-like bound}
        Let $\mathcal{C}\subseteq\mathbb{F}_{q^m}^\ell$ be a rank-metric code with $\lvert\mathcal{C}\rvert\geq 2$ and minimum distance $d$. Then
        \[
            \lvert\mathcal{C}\rvert\leq q^{\max\{m,\ell\}(\min\{m,\ell\}-d+1)}.
        \]
        If the parameters of $\mathcal{C}$ meet the bound, then the code is called a Maximum Rank Distance (MRD) code.
    \end{theorem}
    \medskip
    As mentioned before, rank-metric codes can be also described in terms of linearized polynomials, however, before introducing them in this setting, we first provide a brief overview of linearized polynomials.
    \begin{definition}
        Let $\mathbb{F}_{q^m}$ be an extension field of $\mathbb{F}_q$ of degree $m$ and let us consider $s\in\{1,\dots,m\}$ with $\mathrm{gcd}(s,m)=1$. For a generator $\sigma\colon x\mapsto x^{q^s}$ of $\mathrm{Gal}(\mathbb{F}_{q^m}/\mathbb{F}_{q})$, a $\sigma$-polynomial (or linearized $\sigma$-polynomial) is a polynomial of the form
        \[
        f(x)\coloneqq\displaystyle\sum_{i=0}^rf_ix^{\sigma^i},
        \]
        where $f_i\in\mathbb{F}_{q^m}$ and $r$ is a non negative integer. If $f_r\neq 0$, $r$ is said to be the $\sigma$-degree of $f(x)$ and we will write $\mathrm{deg}_{\sigma}(f(x))=r$. We will denote by $\mathbb{L}_{\sigma}(\mathbb{F}_{q^m})$ the set of all $\sigma$-polynomials over $\mathbb{F}_{q^m}$. If $s=1$, they are known as linearized $q$-polynomials, or simply linearized polynomials.
    \end{definition}
    \noindent The name linearized $\sigma$-polynomial comes from their property of being $\mathbb{F}_q$-linear maps from the $\mathbb{F}_q$-vector space $\mathbb{F}_{q^m}$ to itself.  An example of linearized $q$-polynomial is the trace function.
    \begin{example}
        Let $r,n$ be positive integers and let $\mathbb{F}_{q^n}$ be a subfield of $\mathbb{F}_{q^{rn}}$. The trace $\mathrm{Tr}_{q^{rn}/q^n}(a)$ of $a\in\mathbb{F}_{q^{rn}}$ over $\mathbb{F}_{q^n}$ is defined by
        \[
        \mathrm{Tr}_{q^{rn}/q^n}(a)=a+a^{q^n}+a^{q^{2n}}+\dots+a^{q^{(r-1)n}}.
        \]
        The trace function $\mathrm{Tr}_{q^{rn}/q^n}$ is an $\mathbb{F}_{q^n}$-linear surjective map from $\mathbb{F}_{q^{rn}}$ to $\mathbb{F}_{q^n}$. For further details we refer to \cite{lidl_finite_fields}.
    \end{example}
    \noindent The set of linearized $\sigma$-polynomials $\mathbb{L}_{\sigma}(\mathbb{F}_{q^m})$, with the usual addition and composition between polynomials, forms a non-commutative $\mathbb{F}_q$-algebra where the scalar multiplication is
    \[
        \cdot\colon(\alpha,f)\in\mathbb{F}_q\times\mathbb{L}_{\sigma}(\mathbb{F}_{q^m})\mapsto\alpha\cdot f\in\mathbb{L}_{\sigma}(\mathbb{F}_{q^m}).
    \]
    Since $(x^{\sigma^m}-x)$ is a two-sided ideal, we can consider the following quotient algebra consisting of linearized polynomials of $\sigma$-degree strictly less than $m$
    \[
    \mathcal{L}_{\sigma,m}(\mathbb{F}_{q^m})\coloneqq\frac{\mathbb{L}_{\sigma}(\mathbb{F}_{q^m})}{(x^{\sigma^m}-x)}=\left\{\displaystyle\sum_{i=0}^{m-1}f_ix^{\sigma^i},~f_i\in\mathbb{F}_{q^m}\right\}.
    \]
    If $s=1$ we simply write $\mathcal{L}_{m}(\mathbb{F}_{q^m})$. It is well-known (e.g. \cite{linearized_polynomials_over_finite_fields}), that the $\mathbb{F}_q$-algebras $(\mathcal{L}_{\sigma,m}(\mathbb{F}_{q^m}),+,\circ,\cdot)$ and $(\mathrm{End}(\mathbb{F}_{q^m}),+,\circ,\cdot)$ are isomorphic thanks to the map that associates to each linearized $\sigma$-polynomial $f$ the endomorphism of $\mathbb{F}_{q^m}$ 
    \[
    a\mapsto\sum_{i=0}^{m-1}f_i\sigma^i(a).
    \]
    As an immediate consequence, the \emph{kernel} and \emph{rank} of a $\sigma$-polynomial can be defined as the kernel and the rank of the corresponding endomorphism and they will be denoted respectively by $\ker(f(x))$ and $\mathrm{rank}(f(x))$. Another consequence is that the $\mathbb{F}_q$-algebra $(\mathcal{L}_{\sigma,m}(\mathbb{F}_{q^m}),+,\circ,\cdot)$ is also isomorphic to the $\mathbb{F}_q$-algebra $\mathbb{F}_q^{m\times m}$ as $\mathbb{F}_{q^m}$ is an $\mathbb{F}_q$-vector space of dimension $m$.

    \noindent A significant result (see \cite[Theorem~5]{Galois_extensions_and_subspaces_of_alternating_bilinear_forms_with_special_rank_properties}, \cite[Theorem~10]{Galois_theory_and_linear_algebra}) concerning linearized polynomials is the following bound on the number of roots, which led to the construction of several families of MRD codes \cite{sheekey_twistedGC,TZ-codes_anewfamily}. Before stating it, recall that the norm function is a map from $\mathbb{F}_{q^m}$ to $\mathbb{F}_q$ such that for all $a\in\mathbb{F}_{q^m}$, 
    \[
    \mathrm{N}_{q^m/q}(a)=a^{1+q+\dots+q^{m-1}}=a^{\frac{q^m-1}{q-1}}.
    \]
    \begin{theorem}
        Consider 
        \[
        f(x)=f_0x+f_1x^{\sigma}+\dots+f_{k-1}x^{\sigma^{k-1}}+f_{k}x^{\sigma^{k}}\in\mathcal{L}_{\sigma,m}(\mathbb{F}_{q^m})
        \]
        and let $f_0,f_1,\dots,f_{k}$ be elements of $\mathbb{F}_{q^m}$ not all of them equal to zero. Then
        \[
            \mathrm{dim}_q(\mathrm{ker}(f))\leq k.
        \]
        Furthermore, if $\mathrm{dim}_q(\mathrm{ker}(f))=k$ then $f_k\neq 0$ and $\mathrm{N}_{q^m/q}(f_0)=(-1)^{mk}\mathrm{N}_{q^m/q}(f_k)$.
    \end{theorem}
    \noindent For more details on linearized polynomials, we refer to \cite{lidl_finite_fields}. 
    \medskip
    \\
    In light of the aforementioned remarks, rank-metric codes can be also regarded as a suitable $\mathbb{F}_q$-subspaces of $\mathcal{L}_{\sigma,m}(\mathbb{F}_{q^m})$. 
    \begin{definition}
        An $\mathbb{F}_q$-linear rank-metric code $\mathcal{C}$ in $\mathcal{L}_{\sigma,m}(\mathbb{F}_{q^m})$ is an $\mathbb{F}_q$-subspace of $\mathcal{L}_{\sigma,m}(\mathbb{F}_{q^m})$ endowed with the rank-metric, i.e., for all $f(x),g(x)\in\mathcal{L}_{\sigma,m}(\mathbb{F}_{q^m})$, 
        \[
            d(f(x),g(x))\coloneqq\mathrm{rank}(f-g)=\mathrm{dim}_q(\mathrm{Im}(f-g)).
        \]
    \end{definition}
    \noindent Since we will need it later in reference to Trombetti-Zhou codes, recall that two $\mathbb{F}_q$-linear rank-metric codes $\mathcal{C},\mathcal{D}\subseteq\mathcal{L}_{\sigma,m}(\mathbb{F}_{q^m})$ are equivalent if there exist two invertible (with respect to $\circ$) $\sigma$-polynomials $h(x),g(x)\in\mathcal{L}_{\sigma,m}(\mathbb{F}_{q^m})$ and a field automorphism $\rho\in\mathrm{Aut}(\mathbb{F}_{q})$ such that
    \[
        \mathcal{D}=h\circ\mathcal{C}^{\rho}\circ g=\{h\circ f^{\rho}\circ g\colon f\in\mathcal{C}\},
    \]
    where $f(x)^{\rho}=\displaystyle\sum_{i=0}^{m-1}\rho(f_i)x^{\sigma^i}$ if $f(x)=\displaystyle\sum_{i=0}^{m-1}f_ix^{\sigma^i}$, see \cite[Definition 4.1]{generalized_TGC}.
    \medskip
    \\
    The notion of Delsarte dual code can be expressed in terms of $\sigma$-polynomials as outlined in \cite[Section 2]{generalized_TGC}. Let us consider the following bilinear form
    \begin{align*}
        \mathit{b}\colon\mathcal{L}_{\sigma,m}(\mathbb{F}_{q^m})\times\mathcal{L}_{\sigma,m}(\mathbb{F}_{q^m})&\to\mathbb{F}_q\\
        (f(x),g(x))&\mapsto\mathit{b}(f(x),g(x))\coloneqq\mathrm{Tr}_{q^m/q}\left(\displaystyle\sum_{i=0}^{m-1}f_ig_i\right),
    \end{align*}
    where $f(x)=\displaystyle\sum_{i=0}^{m-1}f_ix^{\sigma^i}$ and $g(x)=\displaystyle\sum_{i=0}^{m-1}g_ix^{\sigma^i}$. 
    \begin{definition}
        \label{def: Delsarte dual code in polynomial framework}
        Let $\mathcal{C}\subseteq\mathcal{L}_{\sigma,m}(\mathbb{F}_{q^m})$ be an $\mathbb{F}_q$-linear rank-metric code. The \textit{Delsarte dual code} $\mathcal{C}^{{\perp}}$ of $\mathcal{C}$ is
    \[
        \mathcal{C}^{{\perp}}=\{f(x)\in\mathcal{L}_{\sigma,m}(\mathbb{F}_{q^m})\colon\mathit{b}(f(x),g(x))=0,\text{ for all }g(x)\in\mathcal{C}\}.
    \]
    \end{definition}
    We can now outline the link between the vector and the linearized polynomial frameworks of rank-metric codes.
    \begin{proposition}
        Let $\underline{\beta}=(\beta_0,\beta_1,\dots,\beta_{m-1})$ be an $\mathbb{F}_q$-basis of $\mathbb{F}_{q^m}$. The map
    \begin{align*}
        \mathit{ev}_{\underline{\beta}}\colon \mathcal{L}_{\sigma,m}(\mathbb{F}_{q^m})&\to\mathbb{F}_{q^m}^m\\
        f(x)&\mapsto\mathit{ev}_{\underline{\beta}}(f(x))\coloneqq(f(\beta_0),f(\beta_1),\dots,f(\beta_{m-1}))
    \end{align*}
    is an $\mathbb{F}_{q^m}$-linear isomorphism which preserves the rank, i.e., $\mathrm{rank}(f(x))=w(\mathit{ev}_{\underline{\beta}}(f(x)))$.
    \end{proposition}
    Furthermore, it is possible to evaluate a code $\mathcal{C}$ in less than $m$ elements. This is due to the fact that any rank-metric code of $\mathcal{L}_{\sigma,m}(\mathbb{F}_{q^m})$ defines a rank-metric code in $\mathbb{F}_{q^m}^\ell$ thanks to the following result presented in \cite[Lemma 2.2]{OnTheListDecodabilityOfRank-metricCodesContainingGabidulinCodes}.
    \begin{lemma}
        \label{remark: it is possible to evaluate a code C in less than m elements}
        Let $\mathcal{C}$ be a rank-metric code in $\mathcal{L}_{\sigma,m}(\mathbb{F}_{q^m})$ with $\lvert\mathcal{C}\rvert=q^{m\kappa}$. Let $\ell$ be a positive integer such that $\kappa \leq \ell \leq m$ and let $S=\{\alpha_0,\alpha_1,\dots,\alpha_{\ell-1}\}$ be a set of $\ell$ $\mathbb{F}_q$-linearly independent elements of $\mathbb{F}_{q^m}$. Assume that $d(\mathcal{C})=m-h$, where $h$ is a positive integer such that $\kappa-1\leq h<\ell$. Then the rank-metric code
        \[
            \tilde{\mathcal{C}}=\{(g(\alpha_0),g(\alpha_1),\dots, g(\alpha_{\ell-1}))\colon g\in\mathcal{C}\} \subseteq \mathbb{F}_{q^m}^{\ell}
        \]
        is a code of $\mathbb{F}_{q^m}^\ell$ with $\lvert\tilde{\mathcal{C}}\rvert=q^{m\kappa}$ and $\ell-h \leq d(\tilde{\mathcal{C}}) \leq \ell-\kappa+1$. In particular, if $\mathcal{C}$ is an MRD code then $\tilde{\mathcal{C}}$ is an MRD code, that is $d(\tilde{\mathcal{C}})=\ell-\kappa+1$.
    \end{lemma}
    This section ends with some useful tools that will be employed throughout the paper.
    \begin{definition}
        Given a vector $\underline{a}=(a_0,a_1,\dots,a_{\ell-1})\in\mathbb{F}_{q^m}^\ell$, the $s\times \ell$ $q$-Vandermonde (or Moore) matrix associated with $\underline{a}$ is given by
        \[
        \mathrm{qvan}_s(\underline{a})\coloneqq\begin{pNiceMatrix}
                a_0 & a_1 & \dots  & a_{\ell-1} \\
                a_0^q & a_1^q & \dots  & a_{\ell-1}^q \\
                \vdots  & \vdots & \ddots & \vdots \\
                a_0^{q^{s-1}} & a_1^{q^{s-1}} & \dots  & a_{\ell-1}^{q^{s-1}}
            \end{pNiceMatrix}.
        \]
    \end{definition}
    \begin{lemma}
        \label{lemma: determinant of q-Vandermonde matrix}
        Let $\underline{a}=(a_0,a_1,\dots,a_{\ell-1})\in\mathbb{F}_{q^m}^\ell$. Then, the determinant of the square $\ell\times \ell$ $q$-Vandermonde matrix associated with $\underline{a}$ is 
        \[
            \mathrm{det}(\mathrm{qvan}_\ell(\underline{a}))=a_0\displaystyle\prod_{j=0}^{\ell-2}\displaystyle\prod_{c_1,\dots,c_j\in\mathbb{F}_{q}}\left(a_{j+1}-\displaystyle\sum_{i=0}^jc_ia_i\right).
        \]
        Hence $\mathrm{det}(\mathrm{qvan}_\ell(\underline{a}))\neq 0$ if and only if $a_0,a_1,\dots,a_{\ell-1}$ are linearly independent over $\mathbb{F}_{q}$. Moreover if $a_0,a_1,\dots,a_{\ell-1}$ are linearly independent over $\mathbb{F}_{q}$, then $\mathrm{qvan}_s(\underline{a})$ has rank $\mathrm{min}\{s,\ell\}$. 
    \end{lemma}
    \section{Overview and new insights into Trombetti-Zhou codes}
    \label{sec: Overview and new insights into Trombetti-Zhou codes}
    Trombetti and Zhou in \cite{TZ-codes_anewfamily} introduced a new family of MRD codes in $\mathbb{F}_q^{2n\times 2n}$ of minimum distance $2\leq d\leq 2n$. Further, by determining their \textit{middle and right nuclei}, they showed that their construction is not equivalent to generalized Gabidulin codes and generalized twisted Gabidulin codes \cite[Corollary 8]{TZ-codes_anewfamily} and to those MRD codes associated with maximum scattered linear sets in \cite[p. 649]{longobardizanellascatt}. 
        \begin{definition}
            Let $n,k,s$ be positive integers satisfying $\mathrm{gcd}(s,2n)=1$, let $q$ be an odd prime power and let $\gamma\in\mathbb{F}_{q^{2n}}$ be such that $\mathrm{N}_{q^{2n}/q}(\gamma)$ is a non-square element in $\mathbb{F}_q$. For a generator $\sigma\colon x\mapsto x^{q^s}$ of $\mathrm{Gal}(\mathbb{F}_{q^{2n}}/\mathbb{F}_q)$, a Trombetti-Zhou code (shortly TZ-code) is the set
            \[
            \mathcal{D}_{k,s}(\gamma)=\left\{ax+\sum_{i=1}^{k-1} f_i x^{\sigma^i}+\gamma bx^{\sigma^k}\colon f_i\in\mathbb{F}_{q^{2n}},~a,b\in \mathbb{F}_{q^n}\right\}\subseteq\mathcal{L}_{\sigma,k+1}(\mathbb{F}_{q^{2n}}).
            \]
        \end{definition}
        Note that the code is defined over $\mathbb{F}_{q^{2n}}$ but it is linear over $\mathbb{F}_{q^n}$. The first and the last coefficients of the polynomials in $\mathcal{D}_{k,s}(\gamma)$ are chosen independently from the field $\mathbb{F}_{q^n}$. If $q$ is even, all the elements of $\mathbb{F}_q$ are square elements, so TZ-codes exist only when the field characteristic is odd. Furthermore, as previously noted, TZ-codes are MRD codes, as proved in \cite[Theorem 3]{TZ-codes_anewfamily}. 
        \begin{theorem}
            A Trombetti-Zhou code $\mathcal{D}_{k,s}(\gamma)$ is an $\mathbb{F}_{q^n}$-linear MRD code over $\mathbb{F}_{q^{2n}}$ of size $q^{2nk}$ and minimum rank distance $2n-k+1$.
        \end{theorem}
         Before presenting the Delsarte dual code of $\mathcal{D}_{k,s}(\gamma)$, it is crucial to acknowledge the following point.
         \begin{remark}
            \label{remark: existence of xi}
            There exists $\xi\in\mathbb{F}_{q^{2n}}\setminus\{0\}$ such that $\mathrm{Tr}_{q^{2n}/q^n}(\gamma\xi)=\gamma\xi+(\gamma\xi)^{q^n}=0$. Since $\mathrm{Tr}_{q^{2n}/q^n}\colon\mathbb{F}_{q^{2n}}\to\mathbb{F}_{q^{n}}$ is an $\mathbb{F}_{q^{n}}$-linear map such that $\mathrm{dim}_{q^n}\left(\ker \left(\mathrm{Tr}_{q^{2n}/q^n}\right)\right)=1$, there exists $\eta\in\mathbb{F}_{q^{2n}}\setminus\{0\}$ such that $\ker \left(\mathrm{Tr}_{q^{2n}/q^n}\right)=\left<\eta\right>_{\mathbb{F}_{q^{n}}}$. Therefore, if we define $\xi\coloneqq\frac{\eta}{\gamma}$, then $\eta=\gamma\xi$ and $\mathrm{Tr}_{q^{2n}/q^n}(\gamma\xi)=\mathrm{Tr}_{q^{2n}/q^n}(\eta)=0$.
         \end{remark}
         Let us take into account $\xi\in\mathbb{F}_{q^{2n}}\setminus\{0\}$ such that $\mathrm{Tr}_{q^{2n}/q^n}(\gamma\xi)=\gamma\xi+(\gamma\xi)^{q^n}=0$, then $\mathcal{D}_{k,s}(\gamma)^{\perp}$ equals
         \[
         \Gamma\coloneqq\left\{\xi\gamma \tilde{b}x+\xi\tilde{a}x^{\sigma^k}+\sum_{i=k+1}^{2n-1} \xi\tilde{f}_ix^{\sigma^i}\colon \tilde{a},\tilde{b}\in \mathbb{F}_{q^n},~\tilde{f}_i \in \mathbb{F}_{q^{2n}}\right\}.
         \]
         Indeed for any 
         \[
         f(x)=ax+\sum_{i=1}^{k-1} f_i x^{\sigma^i}+\gamma bx^{\sigma^k}\in\mathcal{D}_{k,s}(\gamma)
         \]
         and 
         \[
         \tilde{f}(x)=\xi\gamma\tilde{b}x+\xi \tilde{a}x^{\sigma^k}+\sum_{i=k+1}^{2n-1} \xi \tilde{f}_ix^{\sigma^i}\in\Gamma,
         \]
         we have
         \begin{align*}
             b(f(x),\tilde{f}(x))&=\mathrm{Tr}_{q^{2n}/q}\left(\gamma\xi a\tilde{b}+\gamma\xi \tilde{a}b\right)\\
             &=\mathrm{Tr}_{q^n/q}\left(\mathrm{Tr}_{q^{2n}/q^n}\left(\gamma\xi a\tilde{b}+\gamma\xi \tilde{a}b\right)\right)\\
             &=\mathrm{Tr}_{q^n/q}\left((a\tilde{b}+\tilde{a}b)\mathrm{Tr}_{q^{2n}/q^n}\left(\gamma\xi\right)\right)\\
             &=\mathrm{Tr}_{q^n/q}(0)=0.
         \end{align*}
         We can say even more thanks to \cite[Proposition 4]{TZ-codes_anewfamily}.
        \begin{proposition}
            The Delsarte dual code of $\mathcal{D}_{k,s}(\gamma)$ is equivalent to $\mathcal{D}_{2n-k,s}(\gamma)$.    
        \end{proposition}
        \noindent For the sake of simplicity, we assume that $s=1$, and denote $\mathcal{D}_{k,1}(\gamma)$ by $\mathcal{D}_{k}(\gamma)$. 
        \begin{remark}
            It is worth noting that $\gamma\in\mathbb{F}_{q^{2n}}\setminus\mathbb{F}_{q^{n}}$. Indeed, if this is not the case, then
         \[
         \mathrm{N}_{q^{2n}/q}(\gamma)=\gamma^{\frac{q^{2n}-1}{q-1}}=\left(\gamma^{\frac{q^{n}-1}{q-1}}\right)^{q^{n}+1}=\left(\gamma^{\frac{q^{n}-1}{q-1}}\right)^2,
         \]
         which leads to a contradiction. Therefore, $\mathbb{F}_{q^{2n}}$ can be seen as a vector space of dimension two over $\mathbb{F}_{q^n}$, i.e., $\mathbb{F}_{q^{2n}}=\mathbb{F}_{q^n}(\gamma)$, and hence
         \[
            \mathcal{D}_{k}(\gamma)=\left<x,x^q,\gamma x^q,x^{q^2},\gamma x^{q^2},\dots,x^{q^{k-1}},\gamma x^{q^{k-1}},\gamma x^{q^k}\right>_{\mathbb{F}_{q^n}}.   
         \] 
        \end{remark}
        \noindent As outlined in Section \ref{sec: Preliminaries}, rank-metric codes can be also regarded as a set of vectors in $\mathbb{F}_{q^{2n}}^{2n}$.
        \begin{proposition}
            Let $\mathcal{D}_{k}(\gamma)$ be a Trombetti-Zhou code and let $\underline{\lambda}=(\lambda_0,\lambda_1,\ldots,\lambda_{2n-1})$ be an $\mathbb{F}_q$-basis of $\mathbb{F}_{q^{2n}}$. The map 
            \begin{align*}
                ev_{\underline{\lambda}}\colon \mathcal{D}_{k}(\gamma)=\left<x,x^q,\gamma x^q,\dots,x^{q^{k-1}},\gamma x^{q^{k-1}},\gamma x^{q^k}\right>_{\mathbb{F}_{q^n}}&\to\mathbb{F}_{q^{2n}}^{2n}\\
                f(x)&\mapsto(f(\lambda_0),f(\lambda_1),\dots,f(\lambda_{2n-1})) \notag
            \end{align*}
            is an $\mathbb{F}_{q^{n}}$-linear injective map which preserves the rank, i.e., $\mathrm{rank}(f(x))=w(ev_{\underline{\lambda}}(f(x)))$.
        \end{proposition}
        \begin{proof}
            Trivially
            \[
            ev_{\underline{\lambda}}(\alpha f(x)+\beta g(x))=\alpha~ev_{\underline{\lambda}}(f(x))+\beta~ev_{\underline{\lambda}}(g(x)),
            \]
            for all $\alpha,\beta\in\mathbb{F}_{q^n}$ and $f(x),g(x)\in\mathcal{D}_k(\gamma)$. Further, $ev_{\underline{\lambda}}$ is injective. Suppose on the contrary $f(\lambda_i)=0$ for all $i\in\{0,1,\ldots,2n-1\}$ and $f\neq 0$. Since $f\in\mathcal{D}_k(\gamma)$, $\mathrm{rank}(f)\geq d=2n-k+1$, therefore $2n\leq\mathrm{dim}_q(\ker (f))<k$, a contradiction. Finally, since $\underline{\lambda}=(\lambda_0,\lambda_1,\ldots,\lambda_{2n-1})$ is an $\mathbb{F}_q$-basis of $\mathbb{F}_{q^{2n}}$, $\{f(\lambda_0),f(\lambda_1),\dots,f(\lambda_{2n-1})\}$ spans $\mathrm{Im}(f)$, therefore
            \begin{align*}
                \mathrm{rank}(f(x))&=\mathrm{dim}_q(\mathrm{Im}(f))\\
                &=\mathrm{dim}_q(\left<f(\lambda_0),f(\lambda_1),\dots,f(\lambda_{2n-1})\right>_{\mathbb{F}_q})\\
                &=w(ev_{\underline{\lambda}}(f(x))).
            \end{align*}
        \end{proof}
        \noindent As a consequence, if we define 
        \[
        \mathcal{TZ}_k(\gamma)[\underline{\lambda}]\coloneqq ev_{\underline{\lambda}}(\mathcal{D}_k(\gamma)),
        \]
        then
        \[
        ev_{\underline{\lambda}}\colon \mathcal{D}_{k}(\gamma)\to\mathcal{TZ}_k(\gamma)[\underline{\lambda}]
        \]
        is an $\mathbb{F}_{q^n}$-linear isomorphism which preserves the rank\footnote{For the sake of simplicity, the same notation $ev_{\underline{\lambda}}$ is employed.}. As a result $\mathcal{TZ}_k(\gamma)[\underline{\lambda}]\subseteq\mathbb{F}_{q^{2n}}^{2n}$ is an $\mathbb{F}_{q^n}$-linear MRD code of size $q^{2nk}$ and minimum rank distance $2n-k+1$, i.e., it is a $q^n\minus[2n,2k,2n-k+1]_{q^{2n}/q}$ code.
        \medskip
        \\
        As a matter of fact, thanks to Lemma \ref{remark: it is possible to evaluate a code C in less than m elements}, it is possible to evaluate a TZ-code $\mathcal{D}_k(\gamma)$ in less than $2n$ elements. For the sake of completeness, we will state and prove Lemma \ref{remark: it is possible to evaluate a code C in less than m elements} in the specific case of TZ-codes.
    \begin{lemma}
        \label{lemma: evaluate in less elements}
        Let $\mathcal{D}_{k}(\gamma)$ be a Trombetti-Zhou code. Let $\ell$ be a positive integer such that $k \leq\ell\leq 2n$ and let $\alpha_0,\alpha_1,\dots,\alpha_{\ell-1}$ be $\mathbb{F}_q$-linearly independent elements of $\mathbb{F}_{q^{2n}}$. Then the rank-metric code
        \[
            \mathcal{TZ}_k(\gamma)[\underline{\alpha}]\coloneqq\{(g(\alpha_0),g(\alpha_1),\dots, g(\alpha_{\ell-1}))\colon g\in\mathcal{D}_{k}(\gamma)\} \subseteq \mathbb{F}_{q^{2n}}^{\ell}
        \]
        is an MRD code of $\mathbb{F}_{q^{2n}}^{\ell}$ with $\left\lvert\mathcal{TZ}_k(\gamma)[\underline{\alpha}]\right\rvert=\left\lvert\mathcal{D}_{k}(\gamma)\right\rvert=q^{2nk}$ and minimum rank distance $\ell-k+1$.
    \end{lemma}
    \begin{proof}
       Let us consider $\mathcal{U}_{\underline{\alpha}}\coloneqq\left<\alpha_0,\alpha_1,\dots,\alpha_{\ell-1}\right>_{\mathbb{F}_q}\subseteq\mathbb{F}_{q^{2n}}$ and let $g\in\mathcal{D}_k(\gamma)$; since $d(\mathcal{D}_k(\gamma))=2n-k+1$, then $\mathrm{dim}_q(\ker (g))\leq k-1$. Now, consider
        \[
        ev_{\underline{\alpha}}\colon g\in\mathcal{D}_k(\gamma)\mapsto(g(\alpha_0),g(\alpha_1),\dots,g(\alpha_{\ell-1}))\in\mathcal{TZ}_k(\gamma)[\underline{\alpha}].
        \]
        First note that $ev_{\underline{\alpha}}$ is injective. Indeed if $g(\alpha_i)=0$ for all $i\in\{0,\dots,\ell-1\}$ and $g\neq 0$, then $\ell\leq\mathrm{dim}_q(\mathrm{ker}(g))<k$, which leads to a contradiction. As $ev_{\underline{\alpha}}$ is also trivially surjective, it can be concluded that it is bijective. So we have that $\left\lvert\mathcal{TZ}_k(\gamma)[\underline{\alpha}]\right\rvert=\left\lvert\mathcal{D}_{k}(\gamma)\right\rvert=q^{2nk}$ and the Singleton bound (Theorem \ref{thm: Singleton-like bound}) implies $d\left(\mathcal{TZ}_k(\gamma)[\underline{\alpha}]\right)\leq \ell-k+1$. Moreover, 
        \[
        w(ev_{\underline{\alpha}}(g))=\ell-\mathrm{dim}_q(\mathrm{ker}(g)\cap\mathcal{U}_{\underline{\alpha}})\geq\ell-k+1,
        \]
        therefore $\mathcal{TZ}_k(\gamma)[\underline{\alpha}]$ is an MRD code with minimum distance $d\left(\mathcal{TZ}_k(\gamma)[\underline{\alpha}]\right)=\ell-k+1$.
    \end{proof}
        In the remainder of this work, we focus on the case $\ell=2n$ and refer the reader to Section \ref{sec: Conclusions and open problems} for the case $k \leq\ell<2n$.
        \medskip
        \\
        It is well known that a linear code can be compactly represented by a generator matrix, whose rows form a basis of the code. However, as previously noted, TZ-codes are $\mathbb{F}_{q^n}$-linear and not $\mathbb{F}_{q^{2n}}$-linear, therefore, one must work on the subfield $\mathbb{F}_{q^n}$ and construct a matrix that plays an analogous role to that of the usual generator matrix. Since $\mathcal{D}_{k}(\gamma)$ is an $\mathbb{F}_{q^n}$-linear space spanned by the $\mathbb{F}_{q^n}$-basis $x,x^q,\gamma x^q,\dots,x^{q^{k-1}},\gamma x^{q^{k-1}},\gamma x^{q^k}$, it follows naturally that an $\mathbb{F}_{q^n}$-basis for $\mathcal{TZ}_k(\gamma)[\underline{\lambda}]$ is given by evaluating $x,x^q,\gamma x^q,\dots,x^{q^{k-1}},\gamma x^{q^{k-1}},\gamma x^{q^k}$ over an $\mathbb{F}_q$-basis of $\mathbb{F}_{q^{2n}}$. If we consider the $2k\times 2n$ matrix $G$ whose rows are an $\mathbb{F}_{q^n}$-basis for $\mathcal{TZ}_k(\gamma)[\underline{\lambda}]$, then   
        \[
        \mathcal{TZ}_k(\gamma)[\underline{\lambda}]=\left\{\underline{x}G\colon \underline{x}\in\mathbb{F}_{q^n}^{2k}\right\}.
        \]
        Because of the latter, we introduce the following more general definition. Henceforth, the term $\mathbb{F}_{q^{n}}$-rank will exclusively connote $\mathbb{F}_{q^{n}}$-rank by rows.
        \begin{definition}
        \label{def: Fqn gen matr}
        Let $n,\kappa,r,\ell$ be positive integers such that $\kappa\leq r\ell$ and let $\mathcal{C}$ be a $q^n\minus[\ell,\kappa]_{q^{rn}/q}$ code. Then, a matrix $G\in\mathbb{F}_{q^{rn}}^{\kappa\times\ell}$ of $\mathbb{F}_{q^{n}}$-rank $\kappa$ is called an $\mathbb{F}_{q^n}$-generator matrix of $\mathcal{C}$ if
        \[
        \mathcal{C}=\left\{\underline{x}G\colon \underline{x}\in\mathbb{F}_{q^n}^{\kappa}\right\}.
        \]
        \end{definition}
        \noindent Although a linear code can be also compactly represented by a parity-check matrix, which has the code as kernel, this does not directly apply to a code that is only linear on a subfield of its whole field of definition. For example, in the case of TZ-codes, a matrix $H\in\mathbb{F}_{q^{2n}}^{(2n-2k)\times 2n}$ such that $GH^{\top}=0$ is not automatically an $\mathbb{F}_{q^n}$-generator matrix of the dual. This is a consequence of the fact that the \emph{natural} orthogonal complement of an $\mathbb{F}_{q^n}$-linear subspace of $\mathbb{F}_{q^{2n}}^{2n}$ is, in general, not an $\mathbb{F}_{q^n}$-linear subspace of $\mathbb{F}_{q^{2n}}^{2n}$. In other words, the inner product \textit{does not preserve $\mathbb{F}_{q^n}$-linearity}. Therefore, the product between a \textit{TZ-code generator matrix} and one of its \textit{TZ-code parity-check matrix} does not have to be the null matrix as in the case of $\mathbb{F}_{q^{2n}}$-linear codes, instead, it has to be a matrix whose entries have trace null over $\mathbb{F}_{q^n}$.
        \par Due to this lack of linearity, an \textit{alternative inner product} must be employed when dealing with the notion of duality.
        \begin{definition}
            \label{def: Fqn dual code}
            Let $n,\kappa,r,\ell$ be positive integers such that $\kappa\leq r\ell$ and let $\mathcal{C}$ be a $q^n\minus[\ell,\kappa]_{q^{rn}/q}$ code. The $\mathbb{F}_{q^n}$-dual code $\mathcal{C}^{\upmodels}$ is an $\mathbb{F}_{q^n}$-linear code of length $\ell$ and $\mathbb{F}_{q^n}$-dimension $r\ell-\kappa$ defined as 
            \[
            \mathcal{C}^{\upmodels}=\left\{\underline{y}\in\mathbb{F}_{q^{rn}}^{\ell}\colon \mathrm{Tr}_{q^{rn}/q^n}\left(\langle\underline{y},\underline{c}\rangle\right)=0\text{ for all }\underline{c}\in\mathcal{C}\right\}.
            \]
        \end{definition}
        \begin{proposition}
            Let $n,\kappa,r,\ell$ be positive integers such that $\kappa\leq r\ell$ and let $\mathcal{C}$ be a $q^n\minus[\ell,\kappa]_{q^{rn}/q}$ code. Consider an $\mathbb{F}_{q^n}$-generator matrix $H\in\mathbb{F}_{q^{rn}}^{(r\ell-\kappa)\times\ell}$ of $\mathcal{C}^{\upmodels}$, then 
            \[
                \mathcal{C}=\left\{\underline{y}\in\mathbb{F}_{q^{rn}}^{\ell}\colon \mathrm{Tr}_{q^{rn}/q^n}\left(\underline{y}H^{\top}\right)=\underline{0}\right\}.
            \]
        \end{proposition}
        
        \begin{proof}
            Recall that $\langle\cdot,\cdot\rangle\colon\mathbb{F}_{q^{rn}}^{\ell}\times\mathbb{F}_{q^{rn}}^{\ell}\to \mathbb{F}_{q^{rn}}$ is a non-degenerate reflexive bilinear form on the $\ell$-dimensional $\mathbb{F}_{q^{rn}}$-vector space $\mathbb{F}_{q^{rn}}^{\ell}$. Let us regard $\mathbb{F}_{q^{rn}}^{\ell}$ as a $rl$-dimensional $\mathbb{F}_{q^n}$-vector space, then the map 
            \[
            (\underline{u},\underline{v})\in\mathbb{F}_{q^{rn}}^{\ell}\times\mathbb{F}_{q^{rn}}^{\ell}\mapsto \mathrm{Tr}_{q^{rn}/q^n}\left(\langle\underline{u},\underline{v}\rangle\right)\in\mathbb{F}_{q^n}
            \]
            turns out to be a non-degenerate reflexive bilinear form on $\mathbb{F}_{q^{rn}}^{\ell}$ seen as an $\mathbb{F}_{q^n}$-vector space (see \cite[p. 9]{PhDthesis_Zullo} for more details). By Definition \ref{def: Fqn dual code},
            \[
            \mathcal{C}\subseteq\left\{\underline{y}\in\mathbb{F}_{q^{rn}}^{\ell}\colon \mathrm{Tr}_{q^{rn}/q^n}\left(\langle\underline{y},\underline{h}\rangle\right)=0\text{ for all }\underline{h}\in\mathrm{Rowspan}_{\mathbb{F}_{q^n}}(H)=\mathcal{C}^{\upmodels}\right\}=\mathrm{Rowspan}_{\mathbb{F}_{q^n}}(H)^{\upmodels}.
            \]
            The only thing left to prove is that an equality holds, which means that $\mathcal{C}$ is the orthogonal complement with respect to $\mathrm{Tr}_{q^{rn}/q^n}\left(\langle\cdot,\cdot\rangle\right)$ of the $\mathbb{F}_{q^n}$-span of the rows of $H$, i.e., $\mathcal{C}^{\upmodels}$. Thanks to \cite[p. 52]{taylor1992geometry}, 
            \[
            \mathrm{dim}_{q^n}\left(\mathrm{Rowspan}_{\mathbb{F}_{q^n}}(H)\right)+ \mathrm{dim}_{q^n}\left(\mathrm{Rowspan}_{\mathbb{F}_{q^n}}(H)^{\upmodels}\right)=\mathrm{dim}_{q^n}\left(\mathbb{F}_{q^{rn}}^{\ell}\right),
            \]
            therefore, since $H\in\mathbb{F}_{q^{rn}}^{(r\ell-\kappa)\times\ell}$ is a matrix of $\mathbb{F}_{q^{n}}$-rank $r\ell-\kappa$, 
            \[
            \mathrm{dim}_{q^n}(\mathcal{C})=\mathrm{dim}_{q^n}\left(\mathrm{Rowspan}_{\mathbb{F}_{q^n}}(H)^{\upmodels}\right)
            \] 
            and hence the equality holds.
        \end{proof}
        \noindent The proposition above leads to the next definition.
        \begin{definition}
        \label{def: Parity Check Matrix Fqn}
        Let $n,\kappa,r,\ell$ be positive integers such that $\kappa\leq r\ell$ and let $\mathcal{C}$ be a $q^n\minus[\ell,\kappa]_{q^{rn}/q}$ code. Then, a matrix $H\in\mathbb{F}_{q^{rn}}^{(r\ell-\kappa)\times\ell}$ of $\mathbb{F}_{q^{n}}$-rank $r\ell-\kappa$ is called an $\mathbb{F}_{q^n}$-parity-check matrix of $\mathcal{C}$ if
        \[
        \mathcal{C}=\left\{\underline{y}\in\mathbb{F}_{q^{rn}}^{\ell}\colon \mathrm{Tr}_{q^{rn}/q^n}\left(\underline{y}H^{\top}\right)=\underline{0}\right\}.
        \]
        For any $\underline{y}\in\mathbb{F}_{q^{rn}}^{\ell}$, we call $\underline{y}H^{\top}$, the $\mathbb{F}_{q^n}$-syndrome of $\underline{y}$ through $H$.
        \end{definition}
        \noindent Similarly to the classical notion of syndrome for linear codes, the $\mathbb{F}_{q^n}$-syndrome also provides a convenient way to determine whether a vector in the ambient space is a codeword. Indeed if $\underline{y}\in\mathbb{F}_{q^{rn}}^{\ell}$, $\mathrm{Tr}_{q^{rn}/q^n}\left(\underline{y}H^{\top}\right)=\underline{0}$ if and only if $\underline{y}\in\mathcal{C}$. As an immediate consequence we also get the following result.
        \begin{proposition}
            \label{prop: Tr(GH^T)=O}
            Let $n,\kappa,r,\ell$ be positive integers such that $\kappa\leq r\ell$ and let $\mathcal{C}$ be a $q^n\minus[\ell,\kappa]_{q^{rn}/q}$ code. Let $G\in\mathbb{F}_{q^{rn}}^{\kappa\times\ell}$ and $H\in\mathbb{F}_{q^{rn}}^{(r\ell-\kappa)\times\ell}$ be, respectively, an $\mathbb{F}_{q^n}$-generator matrix and $\mathbb{F}_{q^n}$-parity-check matrix of the code. Then 
            \[
            \mathrm{Tr}_{q^{rn}/q^n}\left(GH^{\top}\right)=0.
            \]
        \end{proposition}
        \begin{remark}
            Observe that Definitions \ref{def: Fqn dual code}, \ref{def: Fqn gen matr}, and \ref{def: Parity Check Matrix Fqn} generalize the notions of Delsarte duality, generator matrix, and parity-check matrix for a linear rank-metric code, respectively. Indeed, when $r=1$, $\mathcal{C}$ is an $[\ell,\kappa]_{q^n/q}$ code and the trace map coincides with the identity map on $\mathbb{F}_{q^n}$.
        \end{remark}
        Subsequently we will focus again on TZ-codes which is $\kappa=2k$, $r=2$ and $\ell=2n$. \medskip
        \\
        \noindent The construction of a TZ-code $\mathbb{F}_{q^n}$-parity-check matrix necessitates the introduction of what will henceforth be referred to as a \textit{trace almost dual basis}.
        \begin{theorem}
            \label{thm: existence of trace almost dual basis}
            Consider an ordered $\mathbb{F}_q$-basis $\underline{\lambda}=(\lambda_0,\lambda_1,\dots,\lambda_{2n-1})$ of $\mathbb{F}_{q^{2n}}$ and let $\gamma\in\mathbb{F}_{q^{2n}}$ be such that $\mathrm{N}_{q^{2n}/q}(\gamma)$ is a non-square element in $\mathbb{F}_q$. Let $\xi\in\mathbb{F}_{q^{2n}}\setminus\{0\}$ with
            \[
            \mathrm{Tr}_{q^{2n}/q^n}(\gamma\xi)=\gamma\xi+(\gamma\xi)^{q^n}=0.
            \]
            Then there exists a unique ordered $\mathbb{F}_{q}$-basis $\underline{\mu}=(\mu_0,\mu_1,\dots,\mu_{2n-1})$ of $\mathbb{F}_{q^{2n}}$ such that for $i,j\in\{0,1,\dots,2n-1\}$ holds
            \begin{equation}
                \label{eq: property almost dual basis}
                \left<\underline{\lambda}^{q^i},\underline{\mu}^{q^j}\right>=\displaystyle\sum_{s=0}^{2n-1}\lambda_s^{q^i}\mu_s^{q^j}=\begin{cases}
                                                0\quad\text{for }i\neq j\\
                                                \neq 0\quad\text{for }i=j
                                            \end{cases}
            \end{equation}
            and, in particular
            \begin{equation}
                \label{eq: condition kth power almost dual basis}
                \left<\underline{\lambda}^{q^k},\underline{\mu}^{q^k}\right>=\xi.
            \end{equation}
        \end{theorem}
        \begin{proof}
            Remark \ref{remark: existence of xi} implies the existence of a $\xi\in\mathbb{F}_{q^{2n}}\setminus\{0\}$ such that $\mathrm{Tr}_{q^{2n}/q^n}(\gamma\xi)=\gamma\xi+(\gamma\xi)^{q^n}=0$. We need to find a vector that satisfies the following linear system arising from conditions (\ref{eq: property almost dual basis}) and (\ref{eq: condition kth power almost dual basis}): 
            \begin{equation}
                \label{eq: syst to get trace almost dual basis}
                \begin{pNiceMatrix}
                    \lambda_0 & \lambda_1 & \dots & \lambda_{2n-1} \\
                    \lambda_0^q & \lambda_1^q & \dots & \lambda_{2n-1}^q \\
                    \vdots & \vdots & \ddots & \vdots \\
                    \lambda_0^{q^{2n-1}} & \lambda_1^{q^{2n-1}} & \dots & \lambda_{2n-1}^{q^{2n-1}}
                \end{pNiceMatrix}\cdot\begin{pNiceMatrix}
                                        x_0 \\
                                        x_1 \\
                                        \vdots \\
                                        x_{2n-1}
                                       \end{pNiceMatrix}=\begin{pNiceMatrix}
                                       \xi^{q^{2n-k}}\\
                                        0 \\
                                        \vdots \\
                                        0
                                       \end{pNiceMatrix}.
            \end{equation}
            Since (\ref{eq: syst to get trace almost dual basis}) is a linear system of $2n$ equations in $2n$ unknowns and its coefficient matrix is non-singular by being a $q$-Vandermonde matrix, there exists a unique solution $\underline{\mu}=(\mu_0,\mu_1,\dots,\mu_{2n-1})$ to (\ref{eq: syst to get trace almost dual basis}). Hence
            $\left<\underline{\lambda},\underline{\mu}\right>=\xi^{q^{2n-k}}$ and thus
            \[
            \left<\underline{\lambda}^{q^k},\underline{\mu}^{q^k}\right>=\left<\underline{\lambda},\underline{\mu}\right>^{q^k}=\xi^{q^{2n}}=\xi.
            \]
            Moreover, as
            \[
            \left<\underline{\lambda}^{q^h},\underline{\mu}\right>=0
            \]
            for all $h\in\{1,\dots,2n-1\}$, then for all $i,j\in\{0,1,\dots,2n-1\}$ with $i\neq j$ 
            \[
            \left<\underline{\lambda}^{q^{i}},\underline{\mu}^{q^{j}}\right>=\left<\underline{\lambda}^{q^{i-j}},\underline{\mu}\right>=0
            \]
            as $i-j\in\left(\mathbb{Z}/2n\mathbb{Z}\right)\setminus\{0\}$. Finally, we observe that
            \[
            \begin{pNiceMatrix}
                \underline{\lambda} \\
                 \underline{\lambda}^q \\
                \vdots \\
                 \underline{\lambda}^{q^{2n-1}}
            \end{pNiceMatrix}\cdot\begin{pNiceMatrix}
                 \underline{\mu} \\
                \underline{\mu}^q \\
                \vdots \\
                \underline{\mu}^{q^{2n-1}}
            \end{pNiceMatrix}^{\top}=\begin{pNiceMatrix}
                \xi^{q^{2n-k}} & 0 & \dots & 0 \\
                0 & \xi^{q^{2n-k+1}} & \dots & 0 \\
                \vdots & \vdots & \ddots & \vdots \\
                0 & 0 & \dots & \xi^{q^{2n-k-1}}
            \end{pNiceMatrix},
            \]
            which is an invertible matrix; therefore, since 
            \[
            \begin{pNiceMatrix}
                 \underline{\mu} \\
                \underline{\mu}^q \\
                \vdots \\
                \underline{\mu}^{q^{2n-1}}
            \end{pNiceMatrix},
            \]
            is a $q$-Vandermonde invertible matrix, $(\mu_0,\mu_1,\dots,\mu_{2n-1})$ is a basis of $\mathbb{F}_{q^{2n}}$ over $\mathbb{F}_{q}$.
        \end{proof}
        The above result enables us to provide the following definition.
        \begin{definition}
            \label{def: almost dual basis}
            Let $\underline{\lambda}=(\lambda_0,\lambda_1,\dots,\lambda_{2n-1})$ be an ordered $\mathbb{F}_{q}$-basis of $\mathbb{F}_{q^{2n}}$. The unique ordered $\mathbb{F}_{q}$-basis $\underline{\mu}=(\mu_0,\mu_1,\dots,\mu_{2n-1})$ of $\mathbb{F}_{q^{2n}}$ arising from Theorem \ref{thm: existence of trace almost dual basis}, is called \textit{trace almost dual basis} of $\underline{\lambda}$.
        \end{definition}
        Thanks to the new tools just outlined, we can now state and prove the following lemma, and thereby furnish an $\mathbb{F}_{q^n}$-generator and an $\mathbb{F}_{q^n}$-parity-check matrix for TZ-codes.
        \begin{lemma}
            \label{lemma: FqnGenerator and Fqnparitycheck of TZ codes}
            Let $\mathcal{D}_{k}(\gamma)$ be a Trombetti-Zhou code, let $\underline{\lambda}=(\lambda_0,\lambda_1,\dots,\lambda_{2n-1})$ be an ordered $\mathbb{F}_q$-basis  of $\mathbb{F}_{q^{2n}}$ and $\underline{\mu}=(\mu_0,\mu_1,\ldots,\mu_{2n-1})$ be its trace almost dual basis. An $\mathbb{F}_{q^n}$-generator matrix and an $\mathbb{F}_{q^n}$-parity-check matrix of $\mathcal{TZ}_k(\gamma)[\underline{\lambda}]$ are, respectively,
            \begin{equation}
                \label{eq: FqnGenerator and Fqn parity check matrix of TZcodes}
                G=\begin{pNiceMatrix}
                \underline{\lambda}\\
                \underline{\lambda}^q\\
                \gamma\underline{\lambda}^q\\
                \vdots\\
                \underline{\lambda}^{q^{k-1}}\\
                \gamma\underline{\lambda}^{q^{k-1}}\\
                \gamma\underline{\lambda}^{q^{k}}
            \end{pNiceMatrix}\in\mathbb{F}_{q^{2n}}^{2k\times 2n}\quad\text{and}\quad H=\begin{pNiceMatrix}
                \gamma^{q^{2n-k}}\underline{\mu}\\
                \underline{\mu}^{q^{k+1}}\\
                \gamma\underline{\mu}^{q^{k+1}}\\
                \vdots\\
                \underline{\mu}^{q^{2n-1}}\\
                \gamma\underline{\mu}^{q^{2n-1}}\\
                \underline{\mu}^{q^k}
            \end{pNiceMatrix}\in\mathbb{F}_{q^{2n}}^{(4n-2k)\times 2n}.
            \end{equation}
        Therefore
        \[
        \mathcal{TZ}_k(\gamma)[\underline{\lambda}]=\left\{\underline{\tilde{f}}G\in\mathbb{F}_{q^{2n}}^{2n}\colon\underline{\tilde{f}}=(a,f_{1,1},f_{1,2},\dots,f_{k-1,1},f_{k-1,2},b)\in\mathbb{F}_{q^n}^{2k}\right\}.
        \]
        \end{lemma}
        \begin{proof}
            Since 
            \[
                ev_{\underline{\lambda}}\colon \mathcal{D}_{k}(\gamma)\to\mathcal{TZ}_k(\gamma)[\underline{\lambda}]
            \]
            is an $\mathbb{F}_{q^n}$-linear isomorphism, it transforms the $\mathbb{F}_{q^n}$-basis $x,x^q,\gamma x^q,\dots,x^{q^{k-1}},\gamma x^{q^{k-1}},\gamma x^{q^k}$ of $\mathcal{D}_{k}(\gamma)$ into the $\mathbb{F}_{q^n}$-basis $\underline{\lambda},\underline{\lambda}^q,\gamma \underline{\lambda}^q,\dots,\underline{\lambda}^{q^{k-1}},\gamma \underline{\lambda}^{q^{k-1}},\gamma \underline{\lambda}^{q^k}$ of $\mathcal{TZ}_k(\gamma)[\underline{\lambda}]$. Consequently $G\in\mathbb{F}_{q^{2n}}^{2k\times 2n}$ is a matrix of $\mathbb{F}_{q^n}$-rank $2k$ such that 
            \[
            \mathcal{TZ}_k(\gamma)[\underline{\lambda}]=\left\{\underline{\tilde{f}}G\in\mathbb{F}_{q^{2n}}^{2n}\colon\underline{\tilde{f}}=(a,f_{1,1},f_{1,2},\dots,f_{k-1,1},f_{k-1,2},b)\in\mathbb{F}_{q^n}^{2k}\right\}.
            \]
            \par In order to proceed with the second part of the proof, let $\xi\in\mathbb{F}_{q^{2n}}\setminus\{0\}$ from Theorem \ref{thm: existence of trace almost dual basis} be the one associated to the trace almost dual basis $\underline{\mu}$. Note that, as $\underline{\mu}$ is an $\mathbb{F}_q$-basis of $\mathbb{F}_{q^{2n}}$, $\underline{\mu},\underline{\mu}^q,\ldots,\underline{\mu}^{q^{2n-1}}$ are $\mathbb{F}_{q^{2n}}$-linearly independent vectors in $\mathbb{F}_{q^{2n}}^{2n}$, hence $H$ has $\mathbb{F}_{q^n}$-rank $4n-2k$. Indeed, if we suppose by contradiction that there exist $b_0,a_1,b_1,\ldots,a_{2n-k-1},b_{2n-k-1},a_{4n-2k-1}\in\mathbb{F}_{q^n}$ not all zero such that 
            \begin{align*}
                \underline{0}&=b_0\gamma^{q^{2n-k}}\underline{\mu}+a_1\underline{\mu}^{q^{k+1}}+b_1\gamma\underline{\mu}^{q^{k+1}}+\ldots+a_{2n-k-1}\underline{\mu}^{q^{2n-1}}+b_{2n-k-1}\gamma\underline{\mu}^{q^{2n-1}}+a_{4n-2k-1}\underline{\mu}^{q^k}\\
                &=b_0\gamma^{q^{2n-k}}\underline{\mu}+c_1\underline{\mu}^{q^{k+1}}+\ldots+c_{2n-k-1}\underline{\mu}^{q^{2n-1}}+a_{4n-2k-1}\underline{\mu}^{q^k},
            \end{align*}
            where $c_i=a_i+\gamma b_i\in\mathbb{F}_{q^{2n}}$ for all $i\in\{1,\ldots,2n-k-1\}$, then we immediately get a contradiction as $\underline{\mu},\underline{\mu}^q,\ldots,\underline{\mu}^{q^{2n-1}}$ are $\mathbb{F}_{q^{2n}}$-linearly independent vectors in $\mathbb{F}_{q^{2n}}^{2n}$. Due to Definition \ref{def: Parity Check Matrix Fqn}, the only thing left to prove is that for any $\underline{\tilde{f}}\in\mathbb{F}_{q^n}^{2k}$, the codeword $\underline{c}=\underline{\tilde{f}}G$ is such that $\mathrm{Tr}_{q^{2n}/q^n}(\underline{c}H^{\top})=\underline{0}$. Trivially,
            
        \[
                \mathrm{Tr}_{q^{2n}/q^n}(\underline{\tilde{f}}GH^{\top})=\left(a\mathrm{Tr}_{q^{2n}/q^n}\left((\gamma\xi)^{q^{2n-k}}\right),0,\dots,0,b\mathrm{Tr}_{q^{2n}/q^n}(\gamma\xi)\right)=\underline{0}
            \]
        due to the choice of $\xi$, i.e., $\mathrm{Tr}_{q^{2n}/q^n}(\gamma\xi)=\mathrm{Tr}_{q^{2n}/q^n}\left((\gamma\xi)^{q^{2n-k}}\right)=0$. 
        \end{proof}
        \begin{remark}
        Note that
        \begin{equation}
            GH^{\top}= 
            \begin{pNiceMatrix}
                (\gamma\xi)^{q^{2n-k}} & 0 & \dots & 0\\
                0 & 0 & \dots & 0\\
                \vdots & \vdots & \ddots & \vdots \\
                0 & 0 & \dots & \gamma\xi\\
            \end{pNiceMatrix},
        \end{equation}
        is not far from being the null matrix and, in accordance with Proposition \ref{prop: Tr(GH^T)=O}, $\mathrm{Tr}_{q^{2n}/q^n}(GH^{\top})=0$.
        \end{remark}
        As an immediate consequence, one can observe that an $\mathbb{F}_{q^n}$-parity-check matrix $H\in\mathbb{F}_{q^{2n}}^{(4n-2k)\times 2n}$ of a $\mathcal{TZ}_k(\gamma)[\underline{\lambda}]$ code is an $\mathbb{F}_{q^n}$-generator matrix of its dual code $\mathcal{TZ}_k(\gamma)[\underline{\lambda}]^{\upmodels}$ with respect to the inner product given by the trace function, Definition \ref{def: Fqn dual code}, and that, such a code, is in turn $ev_{\underline{\mu}}\left(\mathcal{D}_k(\gamma)^{\perp}\right)$, where $\mathcal{D}_k(\gamma)^{\perp}$ is the Delsarte dual code in Definition \ref{def: Delsarte dual code in polynomial framework}. Figure \ref{tikz: diagram} provides a visual representation of the latter consideration.
        \begin{figure}[H]
            \centering
            \begin{tikzpicture}[node distance=5cm, auto]
                \node (A) at (0, 1) {\textbf{\normalsize $\mathcal{D}_k(\gamma)$}}; 
                \node (B) at (6, 1) {\textbf{\normalsize $\mathcal{TZ}_k(\gamma)[\underline{\lambda}]$}};
                \node (C) at (0, -1) {\textbf{\normalsize $\mathcal{D}_k(\gamma)^{\perp}$}}; 
                \node (D) at (6, -1) {\textbf{\normalsize $\mathcal{TZ}_k(\gamma)[\underline{\lambda}]^{\upmodels}$}};

                \draw[->] ([yshift=1mm]A.east) -- ([yshift=1mm]B.west) node[midway, above] {$ev_{\underline{\lambda}}$};
                \draw[->] ([yshift=-1mm]B.west) -- ([yshift=-1mm]A.east) node[midway, below] {$ev_{\underline{\lambda}}^{-1}$};

                \draw[->] ([yshift=1mm]C.east) -- ([yshift=1mm]D.west) node[midway, above] {$ev_{\underline{\mu}}$};
                \draw[->] ([yshift=-1mm]D.west) -- ([yshift=-1mm]C.east) node[midway, below] {$ev_{\underline{\mu}}^{-1}$};

                \draw[->] ([xshift=-1mm]A.south) -- ([xshift=-1mm]C.north) node[midway, left] {$\perp$};
                \draw[->] ([xshift=1mm]C.north) -- ([xshift=1mm]A.south) node[midway, right] {$\perp$};

                \draw[->] ([xshift=-1mm]B.south) -- ([xshift=-1mm]D.north) node[midway, left] {$\upmodels$};
                \draw[->] ([xshift=1mm]D.north) -- ([xshift=1mm]B.south) node[midway, right] {$\upmodels$};
            \end{tikzpicture}
            \caption{Relationships among the rank-metric codes $\mathcal{D}_k(\gamma)$,$\mathcal{D}_k(\gamma)^{\perp}$, $\mathcal{TZ}_k(\gamma)[\underline{\lambda}]$ and $\mathcal{TZ}_k(\gamma)[\underline{\lambda}]^{\upmodels}$.} 
            \label{tikz: diagram}
        \end{figure}
        \noindent Despite the fact that a $\mathcal{TZ}_k(\gamma)[\underline{\lambda}]$ code is not $\mathbb{F}_{q^{2n}}$-linear, henceforth the terms $\mathbb{F}_{q^n}$\textit{-generator matrix}, $\mathbb{F}_{q^n}$\textit{-parity-check matrix} and $\mathbb{F}_{q^n}$\textit{-syndrome} shall be occasionally referred to as generator matrix, parity-check matrix and syndrome, in accordance to the role the former will play in the new syndrome-based decoding of TZ-codes. In addition, we will assume that an $\mathbb{F}_q$-basis $\underline{\lambda}$ is fixed and use the following notation for the remainder of this work $\mathcal{TZ}_k(\gamma)\coloneqq\mathcal{TZ}_k(\gamma)[\underline{\lambda}]$.
        \medskip
        \\
        \par To conclude this section, we present a small but illustrative example of a generator matrix and a parity-check matrix for a TZ code.
        \begin{example}
            Let $q=5$, $n=2$ and $k=2$, and let us consider the field $\mathbb{F}_{625}=\mathbb{F}_{5}(\alpha)$ where $\alpha^4+2=0$. Then $\gamma=\alpha^3+\alpha^2+2\alpha+3$ has norm $\mathrm{N}_{625/5}(\gamma)=2$, that is a non-square element in $\mathbb{F}_5$. Consider the $\mathbb{F}_5$-basis $\underline{\lambda}=(1,\alpha,\alpha^2,\alpha^3)$ of $\mathbb{F}_{625}$ and the TZ-code 
            \[
            \mathcal{TZ}_2(\gamma)=\{(f(1),f(\alpha),f(\alpha^2),f(\alpha^3))\colon f\in\mathcal{D}_2(\gamma)\}\subseteq\mathbb{F}_{625}^4,
            \]
            where 
            \[
            \mathcal{D}_2(\gamma)=\{ax+f_1x^5+\gamma bx^{25}\colon f_1\in\mathbb{F}_{625},~ a,b\in\mathbb{F}_{25}\}.
            \]
            Remark \ref{remark: existence of xi}, implies the existence of $\eta=4\alpha^3+\alpha$ such that $\mathrm{ker}\left(\mathrm{Tr}_{625/25}\right)=\left<4\alpha^3+\alpha\right>_{\mathbb{F}_{25}}$, while $\xi$ in Theorem \ref{thm: existence of trace almost dual basis} is equal to
            \[
            \xi=\frac{\eta}{\gamma}=\frac{4\alpha^3+\alpha}{\alpha^3+\alpha^2+2\alpha+3}=4\alpha^2+2\alpha+4.
            \]
            Consequently the trace almost dual basis $\underline{\mu}=(\mu_0,\mu_1,\mu_2,\mu_3)$ of $\underline{\lambda}$ is the solution of the following linear system
            \[
            \begin{pmatrix}
                1 & \alpha & \alpha^2 & \alpha^3\\
                1 & 3\alpha & 4\alpha^2 & 2\alpha^3\\
                1 & 4\alpha & \alpha^2 & 4\alpha^3\\
                1 & 2\alpha & 4\alpha^2 & 3\alpha^3
            \end{pmatrix}\cdot\begin{pmatrix}
               \mu_0\\
               \mu_1\\
               \mu_2\\
               \mu_3
            \end{pmatrix}=\begin{pmatrix}
               4\alpha^2 + 3\alpha + 4\\
               0\\
               0\\
               0
            \end{pmatrix},
            \]
            i.e., $\underline{\mu}=(\alpha^2 + 2\alpha + 1, 2\alpha^3 + \alpha + 2, 4\alpha^3 + 2\alpha^2 + 1, 2\alpha^3 + 4\alpha^2 + 2\alpha)$. Therefore the generator matrix and parity-check matrix of $\mathcal{TZ}_2(\gamma)$ are
            \[
            G=\begin{pmatrix}
                1 & \alpha & \alpha^2 & \alpha^3\\
                1 & 3\alpha & 4\alpha^2 & 2\alpha^3\\
                \alpha^3 + \alpha^2 + 2\alpha + 3 & 3\alpha^3 + \alpha^2 + 4\alpha + 4 & 3\alpha^3 + 2\alpha^2 + 2\alpha + 2 & \alpha^3 + \alpha^2 + \alpha + 2\\
                \alpha^3 + \alpha^2 + 2\alpha + 3 & 4\alpha^3 + 3\alpha^2 + 2\alpha + 2  & 2\alpha^3 + 3\alpha^2 + 3\alpha + 3 & 2\alpha^3 + 2\alpha^2 + 2\alpha + 4
               \end{pmatrix}\in\mathbb{F}_{625}^{4\times4}
            \]
            and
            \[
               H=\begin{pmatrix}
                4\alpha^3 + \alpha & 4\alpha^2 + 1 & 2\alpha^3 + 4\alpha & 2\alpha^2 + 4\\
                4\alpha^2 + 4\alpha + 1 & \alpha^3 + 2\alpha + 2 & 2\alpha^3 + 3\alpha^2 + 1 & \alpha^3 + \alpha^2 + 4\alpha\\
                3\alpha^3 + \alpha^2 + \alpha + 2 & 2\alpha^3 + 4\alpha^2 + 3\alpha + 3 & 3\alpha^3 + \alpha^2 + 2\alpha + 4 & 4\alpha^3 + 4\alpha^2 + 3\alpha + 1\\
                \alpha^2 + 3\alpha + 1 & 3\alpha^3 + 4\alpha + 2  & \alpha^3 + 2\alpha^2 + 1 & 3\alpha^3 + 4\alpha^2 + 3\alpha
               \end{pmatrix}\in\mathbb{F}_{625}^{4\times4}.
            \]
            To conclude, we note that
            \[
            GH^{\top}=\begin{pmatrix}
                \alpha^3 + 4\alpha & 0 & 0 & 0\\
                    0 & 0 & 0 & 0\\
                    0 & 0 & 0 & 0\\
                    0 & 0 & 0 & 4\alpha^3 + \alpha
               \end{pmatrix}
            \]
            is a matrix whose entries have null trace over $\mathbb{F}_{25}$.
        \end{example}
        \section{Decoding Trombetti-Zhou codes: a new syndrome-based decoding approach}
        \label{sec: Decoding of Trombetti-Zhou codes: a new syndrome-based decoding approach}
        In the following, we adapt the well-known syndrome-based decoding for Gabidulin codes to TZ-codes. Our approach follows \cite[Subsection 3.2.1]{PhdThesis_WZeh}, however, the original algorithm was firstly proposed in \cite{gabidulin_TheoryOfCodesWithMaximumRankDistance}.\\
        
        \noindent Before proceeding, it is worth spending some words on the encoding of TZ-codes. Since $\mathcal{TZ}_k(\gamma)$ is an $\mathbb{F}_{q^n}$-linear space of dimension $2k$ spanned by the $\mathbb{F}_{q^n}$-basis $\underline{\lambda},\underline{\lambda}^q,\gamma \underline{\lambda}^q,\dots,\underline{\lambda}^{q^{k-1}},\gamma \underline{\lambda}^{q^{k-1}},\gamma \underline{\lambda}^{q^k}$, information words in $\mathbb{F}_{q^n}^{2k}$ can be encoded into codewords in $\mathcal{TZ}_k(\gamma)\subseteq\mathbb{F}_{q^{2n}}^{2n}$ using an $\mathbb{F}_{q^n}$-generator matrix. In other words, if $G$ is an $\mathbb{F}_{q^n}$-generator matrix of $\mathcal{TZ}_k(\gamma)$, the encoding of a message $\underline{\tilde{f}}=(a,f_{1,1},f_{1,2},\dots,f_{k-1,1},f_{k-1,2},b)\in\mathbb{F}_{q^n}^{2k}$ is given by the following map:
            \[
                \underline{\tilde{f}}=(a,f_{1,1},f_{1,2},\dots,f_{k-1,1},f_{k-1,2},b)\in\mathbb{F}_{q^n}^{2k}\mapsto\underline{c}\coloneqq\underline{\tilde{f}}G\in\mathcal{TZ}_k(\gamma)\subseteq\mathbb{F}_{q^{2n}}^{2n}.
            \]
        Let $\mathcal{D}_{k}(\gamma)$ be a Trombetti-Zhou code and let $G\in\mathbb{F}_{q^{2n}}^{2k\times 2n}$ be a generator matrix of $\mathcal{TZ}_k(\gamma)$ as in Lemma \ref{lemma: FqnGenerator and Fqnparitycheck of TZ codes}. Further, let $\underline{r}=\underline{\tilde{f}}G+\underline{e}\in\mathbb{F}_{q^{2n}}^{2n}$ be the received word, where $\underline{\tilde{f}}\in\mathbb{F}_{q^n}^{2k}$ is the sent message, $\underline{c}=\underline{\tilde{f}}G$ is the encoded message and $\underline{e}$ is the error vector. In this work we are interested in a bounded minimum distance decoding algorithm which guarantees to find either a unique codeword within a radius no greater than $\left\lfloor\frac{d-1}{2}\right\rfloor$ or the empty set. In the latter case, a decoding failure is declared. We refer to \cite[Subsection 2.1.3]{PhdThesis_WZeh} for basics of block codes and decoding principles.\par
        To remain within the hypothesis of unique decoding, our theoretical analysis for designing a decoding algorithm for Trombetti-Zhou codes assumes that the rank weight $t$ of the error vector $\underline{e}$ is 
        \[
        t\leq\left\lfloor\frac{d-1}{2}\right\rfloor=\left\lfloor\frac{2n-k}{2}\right\rfloor
        \]
        and that, when the rank $t$ achieves the unique error correcting radius (that is when $k$ is even and $t=n-\frac{k}{2}$), the entries of the error vector are constrained to the subfield $\mathbb{F}_{q^n}$ of $\mathbb{F}_{q^{2n}}$, i.e., $\underline{e}\in\mathbb{F}_{q^n}^{2n}\subseteq\mathbb{F}_{q^{2n}}^{2n}$. The latter assumption is strictly connected to the $\mathbb{F}_{q^n}$-linearity of TZ-codes and it will allow us to exploit the trace map to reconstruct the error even when the unique decoding radius is met. To sum up, our goal is to find the unique codeword $\underline{c}$ such that $d(\underline{c},\underline{r})\leq \left\lfloor\frac{2n-k}{2}\right\rfloor$.
        \medskip
        \\
        The initial step in the decoding of TZ-codes is the decomposition of the error, which is based on the following rank decomposition theorem.
        \begin{theorem}{\textnormal{\textbf{(Rank decomposition \cite{PhdThesis_WZeh}, Lemma 3.8).}}}
            \label{thm: Rank decomposition}
            For any matrix $X\in\mathbb{F}_q^{m\times h}$ of rank $t$ there exist full rank matrices $Y\in\mathbb{F}_q^{m\times t}$ and $Z\in\mathbb{F}_q^{t\times h}$ such that $X=YZ$. Moreover $\mathrm{Colspan}_{\mathbb{F}_q}(X)=\mathrm{Colspan}_{\mathbb{F}_q}(Y)$ and $\mathrm{Rowspan}_{\mathbb{F}_q}(X)=\mathrm{Rowspan}_{\mathbb{F}_q}(Z)$.
        \end{theorem}
        Therefore, once fixed a basis $\underline{\beta}=(\beta_0,\beta_1,\dots,\beta_{2n-1})$ of $\mathbb{F}_{q^{2n}}$ over $\mathbb{F}_{q}$, if $\mathrm{wt}(\underline{e})=t$, the matrix representation of $\underline{e}$ can be rewritten as follows:
        \[
            E=\mathrm{ext}_{\underline{\beta}}(\underline{e})=AB,
        \]
        where $A\in\mathbb{F}_q^{2n\times t}$, $B\in\mathbb{F}_q^{t\times 2n}$ and $\mathrm{rank}(A)=\mathrm{rank}(B)=t$. If we define $\underline{a}=(a_0,a_1,\ldots,a_{t-1})\coloneqq\mathrm{ext}_{\underline{\beta}}^{-1}(A)\in\mathbb{F}_{q^{2n}}^{t}$, then
        \begin{align}
            \label{eq: rank decomposition of the error vector}
            \underline{e}&=\mathrm{ext}_{\underline{\beta}}^{-1}(E)=\mathrm{ext}_{\underline{\beta}}^{-1}(AB)\notag\\
            &=(\beta_0,\beta_1,\dots,\beta_{2n-1})AB \notag\\
            &=\left(\mathrm{ext}_{\underline{\beta}}^{-1}(A)\right)B \notag \\ 
            &=(a_0,a_1,\dots,a_{t-1})B \notag\\
            &=\left(\displaystyle\sum_{l=0}^{t-1}a_lB_{l,0},\displaystyle\sum_{l=0}^{t-1}a_lB_{l,1},\dots,\displaystyle\sum_{l=0}^{t-1}a_lB_{l,2n-1}\right).
        \end{align}
        This decomposition is not unique, however, each of them is suitable for decoding (see the Appendix) and once $\underline{a}$ is fixed, $B$ is unique. Another key point of the algorithm is the introduction of an intermediate unknown $\underline{d}\in\mathbb{F}_{q^{2n}}^t$, called \textit{vector of error locators}, which satisfies $\underline{d}^{\top}=B\underline{\mu}^{{q^k}^{\top}}$.
        \medskip
        \\
        Algorithm \ref{alg:Decoding of Trombetti-Zhou codes: a new syndrome-based decoding approach.} summarizes the final decoding procedure. However, to understand how it works, it is first necessary to provide a comprehensive illustration of the theoretical foundation upon which it is based. To facilitate the comprehension of the following sections, we also furnish a preliminary toy version of the decoding process that combines theoretical considerations with practical steps of the final algorithm.
        \begin{itemize}
            \renewcommand{\labelitemi}{$\circledast$}
            \item Upon reception of the vector $\underline{r}$, compute the syndrome as $\underline{s}=\underline{r}\cdot H^{\top}$, where $H$ is as in (\ref{eq: FqnGenerator and Fqn parity check matrix of TZcodes}).
            \item If $\mathrm{Tr}_{q^{2n}/q^n}(\underline{s})=\underline{0}$, then $\underline{c}=\underline{r}$, otherwise it is necessary to reconstruct a potential decomposition of the error vector as $\underline{e}=\underline{a}B$.
            \item Identify $\underline{a}$ of a potential decomposition from
            \[
                \mathrm{ker}(\Lambda_0x+\Lambda_1x^q+\ldots+\Lambda_tx^{q^t})=\left<a_0,a_1,\ldots,a_{t-1}\right>_{\mathbb{F}_q},
            \] 
            where the coefficients $\Lambda_0,\Lambda_1,\ldots,\Lambda_t\in\mathbb{F}_{q^{2n}}$ are obtained as solution af a homogeneous linear system whose matrix coefficients involve powers of the entries of the syndrome.
            \item Determine $\underline{d}$ by solving a uniquely solvable linear system arising from $\underline{a}$ and the syndrome $\underline{s}$.
            \item Determine $B$ from $\underline{d}^{\top}=B\underline{\mu}^{{q^k}^{\top}}$
            \item Reconstruct the sent codeword $\underline{c}=\underline{r}-\underline{a}B$.
        \end{itemize}
        \subsection{Syndrome Calculation}\label{sec: Syndrome Calculation} Let $H\in\mathbb{F}_{q^{2n}}^{(4n-2k)\times2n}$ be a parity-check matrix of $\mathcal{TZ}_k(\gamma)$ as in (\ref{eq: FqnGenerator and Fqn parity check matrix of TZcodes}). The first step in decoding TZ-codes is to compute the syndrome of the received word $\underline{r}=\underline{\tilde{f}}G+\underline{e}\in\mathbb{F}_{q^{2n}}^{2n}$, that is
        \[
        \underline{s}=(s_0,s_1,\dots,s_{4n-2k-1})=\underline{r}H^{\top}=(\underline{\tilde{f}}G+\underline{e}) H^{\top}.
        \]
        More explicitly
        \begin{align*}
            \underline{s}&=(s_0,s_1,\dots,s_{4n-2k-1})=(a(\gamma\xi)^{q^{2n-k}},0,\dots,0,b\gamma\xi)+\underline{e}\begin{pNiceMatrix}
                \gamma^{q^{2n-k}}\underline{\mu}\\
                \underline{\mu}^{q^{k+1}}\\
                \gamma\underline{\mu}^{q^{k+1}}\\
                \vdots\\
                \underline{\mu}^{q^{2n-1}}\\
                \gamma\underline{\mu}^{q^{2n-1}}\\
                \underline{\mu}^{q^{k}}
            \end{pNiceMatrix}^{\top}\\
                        &=\left(a(\gamma\xi)^{q^{2n-k}}+\left<\underline{e},\gamma^{q^{2n-k}}\underline{\mu}\right>,\left<\underline{e},\underline{\mu}^{q^{k+1}}\right>,\dots,\gamma\left<\underline{e},\underline{\mu}^{q^{2n-1}}\right>,b\gamma\xi+\left<\underline{e},\underline{\mu}^{q^k}\right>\right).
            \end{align*}
            Keep in mind that this analysis assumes that an error has occurred, therefore $\mathrm{Tr}_{q^{2n}/q^n}(\underline{s})\neq\underline{0}$.
            \medskip
            \\
            In order to highlight the shape of the entries of the syndrome, as already outlined, for all $\ell\in\{0,1,\dots,t-1\}$ we define the \textit{error locators} associated with $B$ to be
            \[
            d_{\ell}\coloneqq d_{\ell}(B)=\displaystyle\sum_{j=0}^{2n-1}B_{\ell,j}\mu_j^{q^k}=\left<\underline{B}_{\ell},\underline{\mu}^{q^k}\right>=\underline{B}_{\ell}\cdot\underline{\mu}^{{q^k}^{\top}}
            \]
            and the \textit{vector of error locators} to be $\underline{d}\coloneqq(d_0,d_1,\dots,d_{t-1})$; therefore
            \begin{equation}
                \label{eqn: vector of error locators}
                    \underline{d}^{\top}=\begin{pmatrix}
                                            d_0\\
                                            d_1\\
                                            \vdots\\
                                            d_{t-1}
                                            \end{pmatrix}=B\underline{\mu}^{{q^k}^{\top}}.
            \end{equation}
            Note that, since $B_{l,j}\in\mathbb{F}_q$, we have that $B_{l,j}^{q^i}=B_{l,j}$ and hence $\underline{d}^{{q^i}^\top}=B\underline{\mu}^{{q^{i+k}}^{\top}}$. 
            \begin{proposition}
            \label{prop: the error locators are linearly independent}
                The error locators $d_0,d_1,\dots,d_{t-1}\in\mathbb{F}_{q^{2n}}$ form a set of linearly independent elements over $\mathbb{F}_q$.
            \end{proposition}
            \begin{proof}
                Note that since $\underline{\mu}$ is an $\mathbb{F}_q$-basis of $\mathbb{F}_{q^{2n}}$, $\underline{\mu}^{q^k}$ is also an $\mathbb{F}_q$-basis of $\mathbb{F}_{q^{2n}}$. Assume that $\{d_0,d_1,\dots,d_{t-1}\}$ are linearly dependent over $\mathbb{F}_q$. By definition there exists $(\delta_0,\delta_1,\dots,\delta_{t-1})\in\mathbb{F}_q^t\setminus \{\underline{0}\}$ such that
                \[
                    \delta_0d_0+\delta_1d_1+\dots+\delta_{t-1}d_{t-1}=0.
                \]
                This happens if and only if 
                \[
                    0=\delta_0\displaystyle\sum_{j=0}^{2n-1}B_{0,j}\mu_j^{q^k}+\delta_1\displaystyle\sum_{j=0}^{2n-1}B_{1,j}\mu_j^{q^k}+\dots+\delta_{t-1}\displaystyle\sum_{j=0}^{2n-1}B_{t-1,j}\mu_j^{q^k}=(\delta_0,\delta_1,\dots,\delta_{t-1})\cdot B\cdot \underline{\mu}^{{q^{k}}^{\top}},
                \]
                which leads to the following homogeneous system
                \[
                    B^{\top}\begin{pNiceMatrix}
                            \delta_0 \\
                            \delta_1 \\
                            \vdots \\
                            \delta_{t-1}
                          \end{pNiceMatrix}=\underline{0},
                \]
                since $\underline{\mu}^{q^k}$ is an $\mathbb{F}_q$-basis of $\mathbb{F}_{q^{2n}}$. The latter is a homogeneous system of rank $t$ in $t$ unknowns, therefore $(\delta_0,\delta_1,\dots,\delta_{t-1})=\underline{0}$, a contradiction. 
            \end{proof}
            With the introduction of the vector of error locators (\ref{eqn: vector of error locators}), the entries of the syndrome are as follows.
            \begin{theorem}
                \label{thm: how we can rewrite the entries of the syndrome thanks to the vector of error locators}
                Let $H\in\mathbb{F}_{q^{2n}}^{(4n-2k)\times2n}$ be a parity-check matrix of $\mathcal{TZ}_k(\gamma)$ as in (\ref{eq: FqnGenerator and Fqn parity check matrix of TZcodes}). Let $\underline{r}=\underline{\tilde{f}}G+\underline{e}\in\mathbb{F}_{q^{2n}}^{2n}$ be the received word and let $\underline{s}$ be the syndrome of $\underline{r}$ through $H$. If $\underline{d}$ is the vector of error locators, then
                \begin{equation}
                    \label{eqn: s_0 coefficient of syndrome polynomial in TZ codes}
                    s_0=a(\gamma\xi)^{q^{2n-k}}+\gamma^{q^{2n-k}}\underline{a}\cdot\underline{d}^{{q^{2n-k}}^{\top}},
                \end{equation}
                \begin{equation}
                    \label{eqn: s_i coefficients of syndrome polynomial in TZ codes}
                    \begin{cases}
                    s_{2i-1}=\underline{a}\cdot\underline{d}^{{q^i}^{\top}}\\
                    s_{2i}=\gamma s_{2i-1}
                \end{cases}
                \end{equation}
                for all $i\in\{1,\dots,2n-k-1\}$, and 
                \begin{equation}
                    \label{eqn: s_last coefficient of syndrome polynomial in TZ codes}
                    s_{4n-2k-1}=b\gamma\xi+\underline{a}\cdot\underline{d}^{\top}.
                \end{equation}
            \end{theorem}
            \begin{proof}
                See the Appendix.
            \end{proof}
            \subsection{Determine The Error Span Polynomial} In order to determine the vector $\underline{a}$ of a potential decomposition, we need to employ the error span polynomial, which is given by the subspace polynomial (of minimum degree) of the $\mathbb{F}_q$-linear subspace of $\mathbb{F}_{q^{2n}}$ generated by the entries of the vector $\underline{a}\in\mathbb{F}_{q^{2n}}^{t}$:
            \begin{equation}
                \label{eq: error span polynomial TZ codes when the error does not reach the unique decoding radius}
                    \Lambda(x)\coloneqq\displaystyle\prod_{\vartheta\in\left<a_0,\dots,a_{t-1}\right>_{\mathbb{F}_q}}\left(x-\vartheta\right).
            \end{equation}
            In accordance with \cite[Theorem 3.52]{lidl_finite_fields}, (\ref{eq: error span polynomial TZ codes when the error does not reach the unique decoding radius}) is a linearized polynomial of $q$-degree $t$ with $\Lambda_i\in\mathbb{F}_{q^{2n}}$ for all $i\in\{0,1,\dots,t\}$, i.e.,
            \[
                \Lambda(x)=\Lambda_0x+\Lambda_1x^{q}+\dots+\Lambda_tx^{q^t}.
            \]
            The coefficients of the error span polynomial can be determined by solving a homogeneous linear system.
            \begin{proposition}
                Let $H\in\mathbb{F}_{q^{2n}}^{(4n-2k)\times2n}$ be a parity-check matrix of $\mathcal{TZ}_k(\gamma)$ as in (\ref{eq: FqnGenerator and Fqn parity check matrix of TZcodes}) and let $\underline{s}$ be the syndrome of the received word $\underline{r}$ through $H$. If $\Lambda_0,\Lambda_1,\ldots,\Lambda_t$ are the coefficients of the error span polynomial in (\ref{eq: error span polynomial TZ codes when the error does not reach the unique decoding radius}), then
                 \begin{equation}
                \label{eq:system of equations to obtain the error span polynomial when the error does not reach the unique decoding radius}
                \begin{pNiceMatrix}
                    s_{2(t+1)-1} & \dots & s_3^{q^{t-1}} & s_1^{q^t}  \\
                    s_{2(t+2)-1} & \dots & s_5^{q^{t-1}} & s_3^{q^t}  \\
                    \vdots & \ddots & \vdots & \vdots \\
                    s_{4n-2k-3} & \dots & s_{4n-2k-1-2t}^{q^{t-1}} & s_{4n-2k-3-2t}^{q^t} 
                    \end{pNiceMatrix}\begin{pNiceMatrix}
                                                    \Lambda_0 \\
                                                    \Lambda_1 \\
                                                    \vdots \\
                                                    \Lambda_t
                                                \end{pNiceMatrix}=\underline{0}^{\top}.
            \end{equation}
            \end{proposition}
            \begin{proof}
                By the definition of error span polynomial (\ref{eq: error span polynomial TZ codes when the error does not reach the unique decoding radius}), for all $\ell\in\{0,1,\dots,t-1\}$
            \begin{equation}
                \label{eq: a_l are roots}
                0=\Lambda(a_{\ell})=\Lambda_0a_{\ell}+\Lambda_1a_{\ell}^{q}+\dots+\Lambda_ta_{\ell}^{q^t}.
            \end{equation}
            Therefore, if we define 
            \[
                \underline{a}_{\ell}\coloneqq(a_{\ell},a_{\ell}^{q},a_{\ell}^{q^2},\dots,a_{\ell}^{q^t})\in\mathbb{F}_{q^{2n}}^{t+1}
            \]
            for all $\ell\in\{0,1,\dots,t-1\}$, (\ref{eq: a_l are roots}) means that
            \[
            (\Lambda_0,\Lambda_1,\Lambda_2,\dots,\Lambda_t)\in\left<\underline{a}_0,\underline{a}_1,\dots,\underline{a}_{t-1}\right>^{\perp}_{\mathbb{F}_{q^{2n}}}.
            \]
            Meanwhile, thanks to (\ref{eqn: s_i coefficients of syndrome polynomial in TZ codes}), we obtain $2n-k-t-1$ systems of the following shape
            \begin{equation}
                \label{eq: link between s_{2(t+j)-1} and vectors a_l}
                \left\{
                \begin{aligned}
                    &s_{2(t+j)-1}=\sum_{\ell=0}^{t-1}a_\ell d_\ell^{q^{t+j}}\\
                    &s_{2(t+j-1)-1}^{q}=\sum_{\ell=0}^{t-1}a_\ell^{q}d_\ell^{q^{t+j}}\\
                    &\quad\vdots \\
                    &s_{2j-1}^{q^t}=\sum_{\ell=0}^{t-1}a_\ell^{q^t}d_\ell^{q^{t+j}},
                \end{aligned}
                \right.
            \end{equation}
            where $j\in\{1,\dots,2n-k-1-t\}$. Hence, if we set
            \[
            \underline{s}_{j}=(s_{2(t+j)-1},s_{2(t+j-1)-1}^{q},\dots,s_{2j-1}^{q^t})
            \]
            for all $j\in\{1,\dots,2n-k-1-t\}$, then
            \begin{align}
                \label{eq: link among vectors s_{2(t+j)-1} and vectors a_l as sum}
                \underline{s}_{j}&=(s_{2(t+j)-1},s_{2(t+j-1)-1}^{q},\dots,s_{2j-1}^{q^t}) \notag \\
                &=\left(\displaystyle\sum_{\ell=0}^{t-1}a_{\ell}d_{\ell}^{q^{t+j}},\displaystyle\sum_{\ell=0}^{t-1}a_{\ell}^{q}d_{\ell}^{q^{t+j}},\dots,\displaystyle\sum_{\ell=0}^{t-1}a_{\ell}^{q^t}d_{\ell}^{q^{t+j}}\right)\notag \\
                &=\displaystyle\sum_{\ell=0}^{t-1}d_{\ell}^{q^{t+j}}(a_{\ell},a_{\ell}^{q},\dots,a_{\ell}^{q^t})=\displaystyle\sum_{\ell=0}^{t-1}d_{\ell}^{q^{t+j}}\underline{a}_{\ell}\notag\\
                &=\begin{pNiceMatrix}
                        d_0^{q^{t+j}} & d_{1}^{q^{t+j}} & \dots & d_{t-1}^{q^{t+j}} \end{pNiceMatrix}\begin{pNiceMatrix}
                                            a_0 & a_{0}^{q} & \dots & a_{0}^{q^t} \\
                                            a_{1} & a_{1}^{q} & \dots & a_{1}^{q^t} \\
                                            \vdots & \vdots & \ddots & \vdots \\
                                            a_{t-1} & a_{t-1}^{q} & \dots & a_{t-1}^{q^t} 
                                        \end{pNiceMatrix}.
            \end{align}
            Consequently,
            \begin{equation}
                \label{eq: condition 1 system to gain error span polynomial coefficients}
                \underline{s}_{1},\dots,\underline{s}_{2n-k-1-t}\in[\underline{a}_0,\underline{a}_1,\dots,\underline{a}_{t-1}]_{\mathbb{F}_{q^{2n}}}=[(\Lambda_0,\Lambda_1,\dots,\Lambda_t)]_{\mathbb{F}_{q^{2n}}}^{\perp},
            \end{equation}
            which leads to the homogeneous linear system of $2n-k-1-t$ equations in $t+1$ unknowns in (\ref{eq:system of equations to obtain the error span polynomial when the error does not reach the unique decoding radius}).
            \end{proof}
            If the dimension of the solution space of (\ref{eq:system of equations to obtain the error span polynomial when the error does not reach the unique decoding radius}) is one, then any solution of (\ref{eq:system of equations to obtain the error span polynomial when the error does not reach the unique decoding radius}) provides the coefficients of the error span polynomial $\Lambda(x)$, as defined in (\ref{eq: error span polynomial TZ codes when the error does not reach the unique decoding radius}), up to a scalar factor. This scalar factor does not present a problem, as it does not alter the root space.\\
            \par
            At this step in the analysis, it is necessary to divide the discussion into two cases, depending on the rank of the error.
            \subsection*{Case 1: $2t + k < 2n$}
            The following lemma provides a criterion for determining the actual rank weight of the error vector in the case where $2t + k < 2n$. Our approach is inspired by \cite[Lemma 3.9]{PhdThesis_WZeh}, as it deals with a related scenario.
        \begin{lemma}{\textnormal{\textbf{(Rank of Syndrome Matrix).}}}
            \label{lemma: Rank of Syndrome Matrix}
            Let $\underline{r}=\underline{c}+\underline{e}\in\mathbb{F}_{q^{2n}}^{2n}$ be given, where $\underline{c}\in\mathcal{TZ}_k(\gamma)$ and $wt(\underline{e})=t$ such that $2t+k<2n$ and let
            \[
            (s_0,s_1,s_2,s_3,\dots,s_{4n-2k-3},s_{4n-2k-2},s_{4n-2k-1})\in\mathbb{F}_{q^{2n}}^{4n-2k}
            \]
            denote the corresponding syndrome. Let $t\leq u\leq\left\lfloor\frac{2n-(k+1)}{2}\right\rfloor$ and let us consider the following $u\times(u+1)$ matrix 
            \begin{equation}
                \label{eq:shape matrix in theorem rank of syndrome matrix}
                S^{(u)}=\begin{pNiceMatrix}
                    s_{2(u+1)-1} & \dots & s_3^{q^{u-1}} & s_1^{q^u}  \\
                    s_{2(u+2)-1} & \dots & s_5^{q^{u-1}} & s_3^{q^u}  \\
                    \vdots & \ddots & \vdots & \vdots \\
                    s_{2(2u)-1} & \dots & s_{2u+1}^{q^{u-1}} & s_{2u-1}^{q^u}
                    \end{pNiceMatrix}
            \end{equation}
            built from the syndrome entries with odd subscript except the last one, i.e., $s_{2j-1}$ for all $j\in\{1,\ldots,2n-k-1\}$. Then $S^{(u)}$ has full rank $u$ if and only if $u=t$.
        \end{lemma}
        \begin{proof}
            With a slight abuse of notation, let us define $a_{\ell},d_{\ell}\coloneqq0$ for $\ell\in\{t,\ldots,u-1\}$. Under this assumption, since $t\leq u\leq\left\lfloor\frac{2n-(k+1)}{2}\right\rfloor$, (\ref{eq: link between s_{2(t+j)-1} and vectors a_l}) can be generalized in the following way
            \[
                \left\{
                    \begin{aligned}
                        &s_{2(u+j)-1}=\sum_{l=0}^{u-1}a_ld_l^{q^{u+j}}\\
                        &s_{2(u+j-1)-1}^{q}=\sum_{l=0}^{u-1}a_l^{q}d_l^{q^{u+j}}\\
                        &\quad\vdots \\
                        &s_{2j-1}^{q^u}=\sum_{l=0}^{u-1}a_l^{q^u}d_l^{q^{u+j}},
                    \end{aligned}
                \right.
            \]
            for all $j\in\{1,\dots,u\}$. Therefore we can decompose $S^{(u)}$ as follows:
            \[
            S^{(u)}=D^{(u)}\cdot A^{(u)}=\begin{pNiceMatrix}
                        d_0^{q^{u+1}} & d_{1}^{q^{u+1}} & \dots & d_{u-1}^{q^{u+1}} \\
                        d_{0}^{q^{u+2}} & d_{1}^{q^{u+2}} & \dots & d_{u-1}^{q^{u+2}} \\
                        \vdots & \vdots & \ddots & \vdots \\
                        d_{0}^{q^{2u}} & d_{1}^{q^{2u}} & \dots & d_{u-1}^{q^{2u}} 
                    \end{pNiceMatrix}\cdot\begin{pNiceMatrix}
                                            a_0 & a_{0}^{q} & \dots & a_{0}^{q^u} \\
                                            a_{1} & a_{1}^{q} & \dots & a_{1}^{q^u} \\
                                            \vdots & \vdots & \ddots & \vdots \\
                                            a_{u-1} & a_{u-1}^{q} & \dots & a_{u-1}^{q^u} 
                                        \end{pNiceMatrix}.
            \]
            We now can observe that both matrices are $q$-Vandermonde matrices therefore they both have full rank if and only if $\{d_0,d_1,\dots,d_{u-1}\}$ and $\{a_0,a_1,\dots,a_{u-1}\}$ are sets of elements which are linearly independent over $\mathbb{F}_q$. If $u>t$, the aforementioned sets are not linearly independent since $a_{\ell},d_{\ell}=0$ for $\ell\in\{t,\ldots,u-1\}$. As a result $\mathrm{rank}(D^{(u)}),\mathrm{rank}(A^{(u)})<u$ and thus $\mathrm{rank}(S^{(u)})\leq \mathrm{min}\{\mathrm{rank}(D^{(u)}),\mathrm{rank}(A^{(u)})\}<u$. On the other side, if $u=t$, $\{d_0,d_1,\dots,d_{t-1}\}$ and $\{a_0,a_1,\dots,a_{t-1}\}$ are sets of elements which are linearly independent over $\mathbb{F}_q$. Moreover $D^{(t)}$ is a square matrix of rank $t$ and $A^{(t)}$ is a $t\times (t+1)$ matrix of rank $t$. The claim now easily follows.
        \end{proof}
        This suggests how to realize one of the first steps of the algorithm to determine the rank of the error: one can set up $S^{(u)}$ for $u=\left\lfloor(2n-(k+1))/2\right\rfloor$ and check its rank. Subsequently, we iteratively decrease $u$ by one until the rank becomes full. The algorithm necessitates the resolution of numerous linear systems of equations over the finite field $\mathbb{F}_{q^{2n}}$ that requires $\mathcal{O}(t^3)$ operations in $\mathbb{F}_{q^{2n}}$ with Gaussian elimination. However, the $q$-circulant structure of the matrix $S^{(u)}$ significantly reduces complexity to $\mathcal{O}(t^2)$ operations by employing an algorithm based on the LEEA from \cite{gabidulin_TheoryOfCodesWithMaximumRankDistance} and the Berlekamp-Massey algorithm from \cite{ModifiedBerlekampMasseyAlgorithm}.
        \subsection{Finding the Root Space of $\boldsymbol{\Lambda(x)}$}\label{subsection: Finding the Root Space} Once the $q$-polynomial $\Lambda(x)$ has been reconstructed, we can determine a possible $\underline{a}=(a_0,a_1,\dots,a_{t-1})$ in the decomposition of (\ref{eq: rank decomposition of the error vector}) as, by construction, it is a basis for the root space of $\Lambda(x)$. In practice, once fixed an $\mathbb{F}_{q}$-basis $\underline{\beta}=\{\beta_0,\beta_1,\dots,\beta_{2n-1}\}$ of $\mathbb{F}_{q^{2n}}$, we have to determine:
        \[
            \mathrm{ker}\left(\mathrm{ext}_{\underline{\beta}}\left(\Lambda(\beta_0),\Lambda(\beta_1),\dots,\Lambda(\beta_{2n-1})\right)\right).
        \]
        In other words we have to solve a linear system of $2n$ equations over $\mathbb{F}_{q}$ (see \cite[Chapter 3, Section 4]{lidl_finite_fields}) which takes at most $\mathcal{O}((2n)^3)$ operations in $\mathbb{F}_q$ to obtain $\mathrm{ext}_{\underline{\beta}}(\underline{a})$ for one possible $\underline{a}$.
        \subsection{Determining the Error}\label{subsection: Determining the error} In order to determine the error $\underline{e}$, it is now necessary to identify the unique matrix $B\in\mathbb{F}_q^{t\times 2n}$ corresponding to the vector $\underline{a}=(a_0,a_1,\dots,a_{t-1})$ such that $\underline{e}=\underline{a}\cdot B$ as in (\ref{eq: rank decomposition of the error vector}). For this we first need to determine $(d_0,d_1,\dots,d_{t-1})$. Taking (\ref{eqn: s_i coefficients of syndrome polynomial in TZ codes}) into account, we can set up the following system of $2n-k-1$ equations (over $\mathbb{F}_{q^{2n}}$) in $t$ unknowns:
        \begin{equation}
            \label{eq:system to obtain d in TZ codes}
            \begin{pNiceMatrix}
                a_0^{q^{-1}} & a_{1}^{q^{-1}} & \dots & a_{t-1}^{q^{-1}} \\
                a_{0}^{q^{-2}} & a_{1}^{q^{-2}} & \dots & a_{t-1}^{q^{-2}} \\
                \vdots & \vdots & \ddots & \vdots \\
                a_{0}^{q^{-(2n-k-1)}} & a_{1}^{q^{-(2n-k-1)}} & \dots & a_{t-1}^{q^{-(2n-k-1)}} 
            \end{pNiceMatrix}\cdot\begin{pNiceMatrix}
                                        d_0 \\
                                        d_1 \\
                                        \vdots \\
                                        d_{t-1}
                                    \end{pNiceMatrix}=\begin{pNiceMatrix}
                                                        s_{1}^{q^{-1}} \\
                                                        s_3^{q^{-2}} \\
                                                        \vdots \\
                                                        s_{4n-2k-3}^{q^{-(2n-k-1)}}
                                                    \end{pNiceMatrix}.
        \end{equation}
        By hypothesis $2n-k>2t$, hence $2n-k-1>2t-1\geq t$, since we are investigating the circumstance in which an error has occurred, i.e., $t\geq 1$. Moreover $a_0,a_1,\dots,a_{t-1}$ are linearly independent; this means that the $2n-k-1\times t$ matrix defined above is of rank $t$ and we can find a unique solution $(d_0,d_1,\dots,d_{t-1})$. Once the vector of error locators $(d_0,d_1,\dots,d_{t-1})$ has been determined, we recover the matrix $B$ from the relations $d_{\ell}=\sum_{j=0}^{2n-1}B_{\ell,j}\mu_j^{q^k}$ for all $\ell\in\{0,\dots,t-1\}$ by simply looking for the representation of $(d_0,d_1,\dots,d_{t-1})$ over $\mathbb{F}_q$ using the $\mathbb{F}_q$-basis $\underline{\mu}^{q^k}$ of $\mathbb{F}_{q^{2n}}$. The computational cost of this step is negligible. Finally, we compute $\underline{e}=\underline{a}\cdot B$ and reconstruct $\underline{c}=\underline{r}-\underline{e}$.
        \subsection*{Case 2: $2t + k = 2n$} When the rank weight $t$ achieves the unique error correcting radius, i.e., when $k$ is even and $t=n-\frac{k}{2}$, more equations are needed to determine the error span polynomial. 
        \begin{proposition}
            Let $\underline{r}=\underline{c}+\underline{e}\in\mathbb{F}_{q^{2n}}^{2n}$ be given, where $\underline{c}\in\mathcal{TZ}_k(\gamma)$ and $\mathrm{rank}(\underline{e})=t$ such that $2n=2t+k$ and let $\underline{s}\in\mathbb{F}_{q^{2n}}^{4n-2k}$ denote the corresponding syndrome. The coefficient matrix of (\ref{eq:system of equations to obtain the error span polynomial when the error does not reach the unique decoding radius}),
        \[
        \widehat{S}^{(t)}=\begin{pNiceMatrix}
                    s_{2(t+1)-1} & \dots & s_3^{q^{t-1}} & s_1^{q^t}  \\
                    s_{2(t+2)-1} & \dots & s_5^{q^{t-1}} & s_3^{q^t}  \\
                    \vdots & \ddots & \vdots & \vdots \\
                    s_{2(2t-1)-1} & \dots & s_{2t-1}^{q^{t-1}} & s_{2(t-1)-1}^{q^t} 
                    \end{pNiceMatrix},
        \]
        has rank $t-1$, therefore the solution space of $\widehat{S}^{(t)}\cdot\underline{\Lambda}^{\top}=\underline{0}^{\top}$ is two-dimensional.
        \end{proposition}
        \begin{proof}
            As it has been previously done in the proof of Lemma \ref{lemma: Rank of Syndrome Matrix}, we can decompose $\widehat{S}^{(t)}$ as
            \[
                \widehat{S}^{(t)}=\widehat{D}^{(t)}\cdot A^{(t)},
            \]
            where
            \[
            \widehat{D}^{(t)}\coloneqq\begin{pNiceMatrix}
                        d_0^{q^{t+1}} & d_{1}^{q^{t+1}} & \dots & d_{t-1}^{q^{t+1}} \\
                        d_0^{q^{t+2}} & d_{1}^{q^{t+2}} & \dots & d_{t-1}^{q^{t+2}} \\
                        \vdots & \vdots & \ddots & \vdots \\
                        d_{0}^{q^{2t-1}} & d_{1}^{q^{2t-1}} & \dots & d_{t-1}^{q^{2t-1}}
                    \end{pNiceMatrix}
            \]
            and
            \[
            A^{(t)}\coloneqq\begin{pNiceMatrix}
                                            a_0 & a_{0}^{q} & \dots & a_{0}^{q^t} \\
                                            a_{1} & a_{1}^{q} & \dots & a_{1}^{q^t} \\
                                            \vdots & \vdots & \ddots & \vdots \\
                                            a_{t-1} & a_{t-1}^{q} & \dots & a_{t-1}^{q^t} 
                                        \end{pNiceMatrix}.
            \]
            We recall that $\{d_0,d_1,\dots,d_{t-1}\}\subseteq\mathbb{F}_{q^{2n}}$ and $\{a_0,a_1,\dots,a_{t-1}\}\subseteq\mathbb{F}_{q^{2n}}$ are sets of elements which are linearly independent over $\mathbb{F}_q$. By Lemma \ref{lemma: determinant of q-Vandermonde matrix}, $A^{(t)}\in\mathbb{F}_{q^{2n}}^{t\times(t+1)}$ and $\widehat{D}^{(t)}\in\mathbb{F}_{q^{2n}}^{(t-1)\times t}$ are matrices of rank $t$ and $t-1$, respectively. Therefore, due to Sylvester's inequality,
            \[
            t-1\geq\mathrm{rank}(\widehat{S}^{(t)})=\mathrm{rank}(\widehat{D}^{(t)}\cdot A^{(t)})\geq \mathrm{rank}(\widehat{D}^{(t)})+\mathrm{rank}(A^{(t)})-t=t-1.
            \]
            In light of the under-determined nature of the system, we will have a set of solutions that has dimension two over $\mathbb{F}_{q^{2n}}$.
        \end{proof}
        To overcome this hurdle, the objective is to incorporate additional information into the system $\widehat{S}^{(t)}\cdot\underline{\Lambda}^{\top}=\underline{0}^{\top}$. Recall the hypothesis that in the case of
            \[
                t=\frac{d-1}{2}=\frac{2n-k}{2},
            \]
        we have $\underline{e}\in\mathbb{F}_{q^{n}}^{2n}\subseteq\mathbb{F}_{q^{2n}}^{2n}$. As a consequence, for a fixed $\mathbb{F}_q$-basis $\underline{\beta}=(\beta_0,\beta_1,\dots,\beta_{n-1})$ of $\mathbb{F}_{q^{n}}$, since
            \[
                wt(\underline{e})=t=n-\frac{k}{2}<n,
            \]
        it follows from Theorem \ref{thm: Rank decomposition} that, $\underline{e}$ can be decomposed as in (\ref{eq: rank decomposition of the error vector}):
        \[
         \underline{e}=\underline{\beta}\cdot AB=(a_0,a_1,\dots,a_{t-1})\cdot B=\underline{a}\cdot B,
        \]
        where $\underline{a}\coloneqq\mathrm{ext}_{\underline{\beta}}^{-1}(A)\in\mathbb{F}_{q^n}^t$, $A\in\mathbb{F}_{q}^{n\times t}$, $B\in\mathbb{F}_{q}^{t\times 2n}$ and $\mathrm{rank}(A)=\mathrm{rank}(B)=t$. Again, there are multiple potential decompositions, yet any of them is suitable for decoding (Proposition \ref{prop: anydecompositionissuitable}). At this stage of the analysis, it is necessary to retrace all the steps of the preceding case. All the considerations made in Section \ref{sec: Syndrome Calculation} remain applicable, with the exception of the fact that $\underline{a}$ is a vector whose entries are constrained in $\mathbb{F}_{q^n}$. To get an overview, if we replace $2n=2t+k$, we have
        \[
            s_0=a(\gamma\xi)^{q^{2t}}+\gamma^{q^{2t}}\underline{a}\cdot\underline{d}^{{q^{2t}}^{\top}},
        \]
        \[
            \begin{cases}
                    s_{2i-1}=\underline{a}\cdot\underline{d}^{{q^i}^{\top}}\\
                    s_{2i}=\gamma s_{2i-1},
            \end{cases}
        \]
        for all $i\in\{1,\dots,2t-3,2t-2,2t-1\}$, and 
        \[
        s_{4t-1}=b\gamma\xi+\underline{a}\cdot\underline{d}^{\top};
        \]
        see (\ref{eqn: s_0 coefficient of syndrome polynomial in TZ codes}), (\ref{eqn: s_i coefficients of syndrome polynomial in TZ codes}) and (\ref{eqn: s_last coefficient of syndrome polynomial in TZ codes}), respectively. \\
        \par Since we are dealing with an $\mathbb{F}_{q^n}$-linear code, we apply the trace function to each entry of the syndrome in order to obtain additional information. For every $j\in\{0,\ldots,4t-1\}$ we define
        \[
        \tilde{s}_{j}\coloneqq\mathrm{Tr}_{q^{2n}/q^n}(s_j).
        \]
        Furthermore, as in Theorem \ref{thm: how we can rewrite the entries of the syndrome thanks to the vector of error locators}, we show how these quantities are related to the vector of error locators $\underline{d}$ and to the vector $\underline{a}$. 
        \begin{theorem}
            \label{thm: obtain more information thanks to the trace}
                Let $H\in\mathbb{F}_{q^{2n}}^{(4n-2k)\times2n}$ be a parity-check matrix of $\mathcal{TZ}_k(\gamma)$ as in (\ref{eq: FqnGenerator and Fqn parity check matrix of TZcodes}). Let $\underline{r}=\underline{\tilde{f}}G+\underline{e}\in\mathbb{F}_{q^{2n}}^{2n}$ be the received word and let $\underline{s}$ be the syndrome of $\underline{r}$ through $H$. If $\underline{d}$ is the vector of error locators, then
                \begin{equation}
                    \label{eqn: trace s_0 coefficient of syndrome polynomial in TZ codes}
                \tilde{s}_0=\underline{a}\cdot \mathrm{Tr}_{q^{2n}/q^n}\left(\gamma^{q^{2t}}\underline{d}^{q^{2t}}\right)^{\top},
                \end{equation}
            \begin{equation}
                \label{eqn: trace s_i coefficients of syndrome polynomial in TZ codes}
                \begin{cases}
                    \tilde{s}_{2i-1}=\underline{a}\cdot \mathrm{Tr}_{q^{2n}/q^n}\left(\underline{d}^{q^i}\right)^{\top}\\
                    \tilde{s}_{2i}=\underline{a}\cdot \mathrm{Tr}_{q^{2n}/q^n}\left(\gamma\underline{d}^{q^i}\right)^{\top},
                \end{cases}
            \end{equation}
            for all $i\in\{1,\dots,2t-1\}$, and 
            \begin{equation}
                \label{eqn: trace s_finale coefficient of syndrome polynomial in TZ codes}
                \tilde{s}_{4t-1}=\underline{a}\cdot \mathrm{Tr}_{q^{2n}/q^n}\left(\underline{d}\right)^{\top}.
            \end{equation}
        \end{theorem}
        \begin{proof}
            See the Appendix.
        \end{proof}
        In order to determine $\underline{a}$, we have to make use of the error span polynomial. However, since $\underline{a}$ is a vector in $\mathbb{F}_{q^n}$, the situation differs slightly from what has previously been investigated.
        \subsection{Determine The Error Span Polynomial for $\mathbf{t=n-\frac{k}{2}}$.} Let
            \begin{equation}
                \label{eq: error span polynomial TZ codes when the error does reach the unique decoding radius}
                    \Lambda(x)\coloneqq\displaystyle\prod_{\vartheta\in\left<a_0,\dots,a_{t-1}\right>_{\mathbb{F}_q}}\left(x-\vartheta\right).
            \end{equation}
            be the error span polynomial. In accordance with \cite[Theorem 3.52]{lidl_finite_fields}, (\ref{eq: error span polynomial TZ codes when the error does reach the unique decoding radius}) is a linearized polynomial $\Lambda(x)=\Lambda_0x+\Lambda_1x^{q}+\dots+\Lambda_tx^{q^t}$ with coefficients $\Lambda_i$ in the subfield of linearity $\mathbb{F}_{q^{n}}$ for all $i\in\{0,1,\dots,t\}$ and, for all $\ell\in\{0,1,\dots,t-1\}$ we have that
            \[
            0=\Lambda(a_{\ell})=\Lambda_0a_{\ell}+\Lambda_1a_{\ell}^{q}+\dots+\Lambda_ta_{\ell}^{q^t}.
            \]
            Therefore, if we define 
            \[
                \underline{a}_{\ell}\coloneqq(a_{\ell},a_{\ell}^{q},a_{\ell}^{q^2},\dots,a_{\ell}^{q^t})\in\mathbb{F}_{q^{n}}^{t+1}
            \]
            for all $\ell\in\{0,1,\dots,t-1\}$, (\ref{eq: a_l are roots}) means that
            \begin{equation}
                \label{eq: perp F_{q^n}}
                (\Lambda_0,\Lambda_1,\Lambda_2,\dots,\Lambda_t)\in[\underline{a}_0,\underline{a}_1,\dots,\underline{a}_{t-1}]^{\perp}_{\mathbb{F}_{q^{n}}}.
            \end{equation}
            Note that since $\underline{a}_0,\underline{a}_1,\dots,\underline{a}_{t-1}\in\mathbb{F}_{q^n}^{t+1}\subseteq\mathbb{F}_{q^{2n}}^{t+1}$ are $t$ linearly independent vectors in $\mathbb{F}_{q^n}^{t+1}$ and hence in $\mathbb{F}_{q^{2n}}^{t+1}$, (\ref{eq: perp F_{q^n}}) leads to the following condition, which will be crucial in the algorithm
            \[
            (\Lambda_0,\Lambda_1,\Lambda_2,\dots,\Lambda_t)\in[\underline{a}_0,\underline{a}_1,\dots,\underline{a}_{t-1}]^{\perp}_{\mathbb{F}_{q^{2n}}}.
            \]
            At this point, since we can construct the same vectors as in (\ref{eq: link among vectors s_{2(t+j)-1} and vectors a_l as sum}) and the considerations in (\ref{eq: condition 1 system to gain error span polynomial coefficients}) still apply, we obtain the system in (\ref{eq:system of equations to obtain the error span polynomial when the error does not reach the unique decoding radius}) in the case where $2n=2t+k$, that is
            \begin{equation}
                \label{eq:system of equations to obtain the error span polynomial when t reaches the unique decoding radius}
                \begin{pNiceMatrix}
                    s_{2(t+1)-1} & \dots & s_3^{q^{t-1}} & s_1^{q^t}  \\
                    s_{2(t+2)-1} & \dots & s_5^{q^{t-1}} & s_3^{q^t}  \\
                    \vdots & \ddots & \vdots & \vdots \\
                    s_{2(2t-1)-1} & \dots & s_{2t-1}^{q^{t-1}} & s_{2(t-1)-1}^{q^t}
                    \end{pNiceMatrix}\begin{pNiceMatrix}
                                                    \Lambda_0 \\
                                                    \Lambda_1 \\
                                                    \vdots \\
                                                    \Lambda_t
                                                \end{pNiceMatrix}=\underline{0}^{\top},
            \end{equation}
            i.e., a system of $t-1$ equations in $t+1$ unknowns. As previously stated, further information is necessary to proceed in order to determine the coefficients of the error span polynomial up to a scalar factor. The $\mathbb{F}_{q^n}$-linearity of the code in exam allows the use of the trace function to add new linearly independent conditions to (\ref{eq:system of equations to obtain the error span polynomial when t reaches the unique decoding radius}). Specifically, the coefficient matrix of (\ref{eq:system of equations to obtain the error span polynomial when t reaches the unique decoding radius}) can be expanded into the following matrix
            \[
                S_{\mathrm{exp}}\coloneqq\begin{pNiceMatrix}
                    s_{2(t+1)-1} & \dots & s_3^{q^{t-1}} & s_1^{q^t}  \\
                    \vdots & \ddots & \vdots & \vdots \\
                    s_{2(2t-1)-1} & \dots & s_{2t-1}^{q^{t-1}} & s_{2(t-1)-1}^{q^t}\\
                    \tilde{s}_{2t-1} &\dots&\tilde{s}_{1}^{q^{t-1}}&\tilde{s}_{4t-1}^{q^t}\\
                    \tilde{s}_{2t+1} & \dots & \tilde{s}_3^{q^{t-1}} & \tilde{s}_1^{q^t}  \\
                    \vdots & \ddots & \vdots & \vdots \\
                    \tilde{s}_{2(2t-1)-1} & \dots & \tilde{s}_{2t-1}^{q^{t-1}} & \tilde{s}_{2(t-1)-1}^{q^t}\\
                    \tilde{s}_0 & \dots & \mathrm{Tr}_{q^{2n}/q^n}\left(\gamma^{q^{2t}}s_{2t-3}^{q^{t-1}}\right) & \mathrm{Tr}_{q^{2n}/q^n}\left(\gamma^{q^{2t}}s_{2t-1}^{q^t}\right)
                    \end{pNiceMatrix},
            \]
            which constitutes the coefficient matrix of a system whose solution space yields the coefficients of the error span polynomial up to a scalar factor even in the case $2t+k=2n$.
            \begin{proposition}
                Let $k$ be even and $t=n-k/2$. Consider a parity-check matrix $H\in\mathbb{F}_{q^{2n}}^{(4n-2k)\times2n}$ of $\mathcal{TZ}_k(\gamma)$ as in (\ref{eq: FqnGenerator and Fqn parity check matrix of TZcodes}) and let $\underline{s}$ be the syndrome of the received word $\underline{r}$ through $H$. If $\Lambda_0,\Lambda_1,\ldots,\Lambda_t$ are the coefficients of the error span polynomial in (\ref{eq: error span polynomial TZ codes when the error does reach the unique decoding radius}), then
                \begin{equation}
                    \label{eq:system of equations to obtain the error span polynomial for TZ codes}
                    S_{\mathrm{exp}}\cdot\underline{\Lambda}^{\top}=\underline{0}^{\top}.
                \end{equation}
            \end{proposition}
            \begin{proof}
                Employing (\ref{eqn: trace s_0 coefficient of syndrome polynomial in TZ codes}), (\ref{eqn: trace s_i coefficients of syndrome polynomial in TZ codes}) and (\ref{eqn: trace s_finale coefficient of syndrome polynomial in TZ codes}), we derive $t-1$ systems that resemble those depicted in (\ref{eq: link between s_{2(t+j)-1} and vectors a_l}):
            \[
                \left\{
                \begin{aligned}
                    &\tilde{s}_{2(t+j)-1}=\underline{a}\cdot \mathrm{Tr}_{q^{2n}/q^n}\left(\underline{d}^{q^{t+j}}\right)^{\top}\\
                    &\tilde{s}_{2(t+j-1)-1}^{q}=\underline{a}^{q}\cdot \mathrm{Tr}_{q^{2n}/q^n}\left(\underline{d}^{q^{t+j}}\right)^{\top}\\
                    &\quad\vdots \\
                    &\tilde{s}_{2j-1}^{q^t}=\underline{a}^{q^t}\cdot \mathrm{Tr}_{q^{2n}/q^n}\left(\underline{d}^{q^{t+j}}\right)^{\top},
                \end{aligned}
                \right.
            \]
            for all $j\in\{1,\dots,t-1\}$. Hence, if we set
            \[
            \underline{\tilde{s}}_j =(\tilde{s}_{2(t+j)-1}, \tilde{s}_{2(t+j-1)-1}^{q}, \dots, \tilde{s}_{2j-1}^{q^t}),
            \]
            for all $j \in \{1,\dots,t-1\}$, then
            \begin{align*}
                \underline{\tilde{s}}_j &= (\tilde{s}_{2(t+j)-1}, \tilde{s}_{2(t+j-1)-1}^{q}, \dots, \tilde{s}_{2j-1}^{q^t}) \notag \\
                &= \left(\sum_{\ell=0}^{t-1} a_{\ell} \left(d_{\ell}^{q^{t+j}} + d_{\ell}^{q^{t+j+n}}\right),\dots, \sum_{\ell=0}^{t-1} a_{\ell}^{q^t} \left(d_{\ell}^{q^{t+j}} + d_{\ell}^{q^{t+j+n}}\right)\right) \notag \\
                &= \sum_{\ell=0}^{t-1} \mathrm{Tr}_{q^{2n}/q^n}\left(d_{\ell}^{q^{t+j}}\right) \underline{a}_{\ell}\notag\\
                &=\begin{pNiceMatrix}
                        \mathrm{Tr}_{q^{2n}/q^n}\left(d_0^{q^{t+j}}\right) & \mathrm{Tr}_{q^{2n}/q^n}\left(d_1^{q^{t+j}}\right) & \dots & \mathrm{Tr}_{q^{2n}/q^n}\left(d_{t-1}^{q^{t+j}}\right) \end{pNiceMatrix}A^{(t)}.
            \end{align*}
            Let us consider now
            \begin{align*}
                \underline{\tilde{s}}_0&=(\tilde{s}_{2t-1},\dots,\tilde{s}_{1}^{q^{t-1}},\tilde{s}_{4t-1}^{q^t})\notag \\
                &=\left(\displaystyle\sum_{\ell=0}^{t-1}a_{\ell}d_{\ell}^{q^t},\dots,\displaystyle\sum_{\ell=0}^{t-1}a_{\ell}^{q^t}d_{\ell}^{q^t}\right)+\left(\displaystyle\sum_{\ell=0}^{t-1}a_{\ell}d_{\ell}^{q^{t+n}},\dots,\displaystyle\sum_{\ell=0}^{t-1}a_{\ell}^{q^t}d_{\ell}^{q^{t+n}}\right)\notag \\
                &=\displaystyle\sum_{\ell=0}^{t-1}\mathrm{Tr}_{q^{2n}/q^n}\left(d_{\ell}^{q^t}\right)\underline{a}_{\ell}\notag\\
                &=\begin{pNiceMatrix}
                        \mathrm{Tr}_{q^{2n}/q^n}\left(d_0^{q^t}\right) & \mathrm{Tr}_{q^{2n}/q^n}\left(d_1^{q^t}\right) & \dots & \mathrm{Tr}_{q^{2n}/q^n}\left(d_{t-1}^{q^t}\right) \end{pNiceMatrix}A^{(t)}
            \end{align*}
            and
            \begin{align*}
                \underline{\bar{s}}_0&=\left(\tilde{s}_0,\mathrm{Tr}_{q^{2n}/q^n}\left(\gamma^{q^{2t}}s_{4t-3}^{q}\right),\dots,\mathrm{Tr}_{q^{2n}/q^n}\left(\gamma^{q^{2t}}s_{2t-1}^{q^t}\right)\right)\notag \\
                &=\left(\displaystyle\sum_{\ell=0}^{t-1}a_{\ell}\mathrm{Tr}_{q^{2n}/q^n}\left(\left(\gamma d_{\ell}\right)^{q^{2t}}\right),\dots,\displaystyle\sum_{\ell=0}^{t-1}a_{\ell}^{q^t}\mathrm{Tr}_{q^{2n}/q^n}\left(\left(\gamma d_{\ell}\right)^{q^{2t}}\right)\right)\notag \\
                &=\displaystyle\sum_{\ell=0}^{t-1}\mathrm{Tr}_{q^{2n}/q^n}\left(\left(\gamma d_{\ell}\right)^{q^{2t}}\right)\underline{a}_{\ell}\notag\\
                &=\begin{pNiceMatrix}
                        \mathrm{Tr}_{q^{2n}/q^n}\left(\gamma d_0^{q^{2t}}\right) & \mathrm{Tr}_{q^{2n}/q^n}\left(\gamma d_1^{q^{2t}}\right) & \dots & \mathrm{Tr}_{q^{2n}/q^n}\left(\gamma d_{t-1}^{q^{2t}}\right) \end{pNiceMatrix}A^{(t)}.
            \end{align*}
            Consequently
            \[
            \underline{\tilde{s}}_{0},\underline{\tilde{s}}_{1},\dots,\underline{\tilde{s}}_{t-1},\underline{\bar{s}}_{0}\in[\underline{a}_0,\underline{a}_1,\dots,\underline{a}_{t-1}]_{\mathbb{F}_{q^{2n}}}=[(\Lambda_0,\Lambda_1,\dots,\Lambda_t)]_{\mathbb{F}_{q^{2n}}}^{\perp}.
            \]
            The latter consideration enables us to expand (\ref{eq:system of equations to obtain the error span polynomial when t reaches the unique decoding radius}) as the system (\ref{eq:system of equations to obtain the error span polynomial for TZ codes}) of $2t$ equations in $t+1$ unknowns.
            \end{proof}
            \begin{theorem}
                \label{corollary: The dimension of the solution space of S_exp is 1}
               The dimension of the solution space of $S_{\mathrm{exp}}\cdot\underline{\Lambda}^{\top}=\underline{0}^{\top}$ is one. 
            \end{theorem}
            \noindent To prove the latter, we will exploit the same procedure used in Lemma \ref{lemma: Rank of Syndrome Matrix}, that is: we decompose the coefficient matrix $S_{\mathrm{exp}}$ of the homogeneous linear system (\ref{eq:system of equations to obtain the error span polynomial for TZ codes}) as
            \begin{equation}
                \label{eq: decomposition of S_exp}
                S_{\mathrm{exp}}=D_{\mathrm{exp}}\cdot A^{(t)},
            \end{equation}
            where
            \begin{equation}
                \label{eq: Dexp}
                D_{\mathrm{exp}}\coloneqq\begin{pNiceMatrix}
                        \underline{d}^{q^{t+1}}\\
                        \vdots\\
                        \underline{d}^{q^{2t-1}}\\
                        \mathrm{Tr}_{q^{2n}/q^n}\left(\underline{d}^{q^t}\right)\\
                        \mathrm{Tr}_{q^{2n}/q^n}\left(\underline{d}^{q^{t+1}}\right)\\
                        \vdots
                        \\
                        \mathrm{Tr}_{q^{2n}/q^n}\left(\underline{d}^{q^{2t-1}}\right)\\
                        \mathrm{Tr}_{q^{2n}/q^n}\left(\gamma^{q^{2t}}\underline{d}^{q^{2t}}\right)
                    \end{pNiceMatrix}\in\mathbb{F}_{q^{2n}}^{2t\times t}
            \end{equation}
            and
            \[
            A^{(t)}=\begin{pNiceMatrix}
                                            a_0 & a_{0}^{q} & \dots & a_{0}^{q^t} \\
                                            a_{1} & a_{1}^{q} & \dots & a_{1}^{q^t} \\
                                            \vdots & \vdots & \ddots & \vdots \\
                                            a_{t-1} & a_{t-1}^{q} & \dots & a_{t-1}^{q^t} 
                                        \end{pNiceMatrix}\in\mathbb{F}_{q^{n}}^{t\times (t+1)}\subseteq\mathbb{F}_{q^{2n}}^{t\times (t+1)}.
            \]
            Since $\{d_0,d_1,\dots,d_{t-1}\}\subseteq\mathbb{F}_{q^{2n}}$ and $\{a_0,a_1,\dots,a_{t-1}\}\subseteq\mathbb{F}_{q^{n}}$ are sets of elements which are linearly independent over $\mathbb{F}_q$, by Lemma \ref{lemma: determinant of q-Vandermonde matrix} $A^{(t)}\in\mathbb{F}_{q^{n}}^{t\times(t+1)}\subseteq\mathbb{F}_{q^{2n}}^{t\times(t+1)}$ is a matrix of rank $t$ over $\mathbb{F}_{q^{2n}}$ while $D_{\mathrm{exp}}\in\mathbb{F}_{q^{2n}}^{2t\times t}$ is a matrix of at least rank $t-1$. We now want to show that $D_{\mathrm{exp}}$ has rank $t$. For ease of notation, in the following theorem we set $[i]\coloneqq q^i$.
            \begin{theorem}
                The matrix $D_{\mathrm{exp}}$ in (\ref{eq: Dexp}) has rank $t$.
            \end{theorem}
            \begin{proof}
            Let us consider the hyperplane 
            \[
            \mathcal{H}\coloneqq\left<\underline{d}^{[t+1]},\ldots,\underline{d}^{[2t-1]}\right>_{\mathbb{F}_{q^{2n}}}
            \]
            in $\mathbb{F}_{q^{2n}}^t$ and let us define
            \[
            \mathfrak{D}\coloneqq\left<\underline{d}^{[t]}+\underline{d}^{[t+n]},\underline{d}^{[t+1]}+\underline{d}^{[t+n+1]},\ldots,\underline{d}^{[2t-1]}+\underline{d}^{[2t+n-1]},\gamma^{[2t]}\underline{d}^{[2t]}+\gamma^{[2t+n]}\underline{d}^{[2t+n]}\right>_{\mathbb{F}_{q^{2n}}}.
            \]
            By contradiction suppose $\mathrm{rank}(D_{\mathrm{exp}})<t$, i.e.,
            \[
            \mathfrak{D}\subseteq\mathcal{H}.
            \]
            The latter is equivalent to
            \begin{align}
                \label{eq: H is also generated by the consecutive n-th powers}
                \mathfrak{D}^{\prime}&\coloneqq\left<\underline{d}^{[t]}+\underline{d}^{[t+n]},\gamma^{[2t]}\underline{d}^{[2t]}+\gamma^{[2t+n]}\underline{d}^{[2t+n]}\right>_{\mathbb{F}_{q^{2n}}}\\
                &\subseteq\mathcal{H}=\left<\underline{d}^{[t+1]},\ldots,\underline{d}^{[2t-1]}\right>_{\mathbb{F}_{q^{2n}}}=\left<\underline{d}^{[t+n+1]},\ldots,\underline{d}^{[2t+n-1]}\right>_{\mathbb{F}_{q^{2n}}}.\notag
            \end{align}
            Let us consider two different cases.
            \medskip
            \\
            \noindent \textbf{Case I:} There exists $\rho\in\mathbb{F}_{q^{2n}}\textnormal{ such that }\underline{d}^{[t]}=\rho\underline{d}^{[t+n]}$.
            \\
            We will subdivide Case I into different cases. However, before addressing them, we note that $\rho$ cannot be zero as the analysis is assuming that an error has occurred.
            \bigskip
            \\
            \textbf{Case I.A:} $\rho\in\mathbb{F}_{q^{2n}}\setminus\left\{-\frac{\gamma^{[t+n]}}{\gamma^{[t]}},-1,0\right\}\textnormal{ and }\underline{d}^{[t]}=\rho\underline{d}^{[t+n]}$.
            \medskip
            \\
            In this scenario, since $\rho\neq -1$, then $1+\frac{1}{\rho}\neq 0$ and hence 
            \[
            \underline{d}^{[t]}+\underline{d}^{[t+n]}=\underline{d}^{[t]}+\frac{1}{\rho}\underline{d}^{[t]}=\left(1+\frac{1}{\rho}\right)\underline{d}^{[t]}.
            \]
            Consequently, 
            \[                      
                \underline{d}^{[t]}\in\mathcal{H}=\left<\underline{d}^{[t+1]},\ldots,\underline{d}^{[2t-1]}\right>_{\mathbb{F}_{q^{2n}}},
            \]
            thus 
            \[
                \mathrm{rank}\begin{pNiceMatrix}
                        \underline{d}^{[t]}\\
                        \underline{d}^{[t+1]}\\
                        \vdots\\
                        \underline{d}^{[2t-1]}
                    \end{pNiceMatrix}=t-1,
            \]
            which leads to a contradiction.
            \medskip
            \\
            \textbf{Case I.B:} $\rho=-1\textnormal{ that is }\mathrm{Tr}_{q^{2n}/q^n}(\underline{d}^{[t]})=\underline{0}$.
            \medskip
            \\
            In this case also $\mathrm{Tr}_{q^{2n}/q^n}(\underline{d}^{[t+i]})=\underline{0}$ for all $i\in\{1,\ldots,t\}$, in particular $\underline{d}^{[2t+n]}=-\underline{d}^{[2t]}$. Therefore 
            \[
                \mathfrak{D}^{\prime}=\left<(\gamma^{[2t+n]}-\gamma^{[2t]})\underline{d}^{[2t]}\right>_{\mathbb{F}_{q^{2n}}}=\left<\underline{d}^{[2t]}\right>_{\mathbb{F}_{q^{2n}}},
            \]
            as $\gamma\notin\mathbb{F}_{q^n}$. Thus,
            \[
            \mathfrak{D}^{\prime}=\left<\underline{d}^{[2t]}\right>_{\mathbb{F}_{q^{2n}}}\subseteq\mathcal{H}=\left<\underline{d}^{[t+1]},\ldots,\underline{d}^{[2t-1]}\right>_{\mathbb{F}_{q^{2n}}},
            \]
            and hence 
            \[
            \mathrm{rank}\begin{pNiceMatrix}
                        \underline{d}^{[t+1]}\\
                        \vdots\\
                        \underline{d}^{[2t-1]}\\
                        \underline{d}^{[2t]}
                    \end{pNiceMatrix}=t-1,
            \]
            a contradiction.
            \medskip
            \\
            \textbf{Case I.C:} $\rho=-\frac{\gamma^{[t+n]}}{\gamma^{[t]}}\textnormal{ that is }\mathrm{Tr}_{q^{2n}/q^n}(\gamma^{[2t]}\underline{d}^{[2t]})=\underline{0}$.
            \medskip
            \\
            In the case in which $\underline{d}^{[t]}=-\frac{\gamma^{[t+n]}}{\gamma^{[t]}}\underline{d}^{[t+n]}$, we have that $\gamma^{[t]}\underline{d}^{[t]}=-\gamma^{[t+n]}\underline{d}^{[t+n]}$ and hence $\gamma^{[2t]}\underline{d}^{[2t]}=-\gamma^{[2t+n]}\underline{d}^{[2t+n]}$, that is $\mathrm{Tr}_{q^{2n}/q^n}(\gamma^{[2t]}\underline{d}^{[2t]})=\underline{0}$. Thus
            \[
            \underline{d}^{[t]}\in\mathcal{H}=\left<\underline{d}^{[t+1]},\ldots,\underline{d}^{[2t-1]}\right>_{\mathbb{F}_{q^{2n}}},
            \]
            and hence 
            \[
            \mathrm{rank}\begin{pNiceMatrix}
                        \underline{d}^{[t]}\\
                        \underline{d}^{[t+1]}\\
                        \vdots\\
                        \underline{d}^{[2t-1]}
                    \end{pNiceMatrix}=t-1,
            \]
            a contradiction. Note that $\frac{\gamma^{[t+n]}}{\gamma^{[t]}}\neq 1$ because $\gamma\notin\mathbb{F}_{q^n}$. The latter is important to note since if $\frac{\gamma^{[t+n]}}{\gamma^{[t]}}$ were one, then $\mathrm{Tr}_{q^{2n}/q^n}(\underline{d}^{[t]})$ would be the null vector.
            \bigskip
            \\
            \textbf{Case II:} $\underline{d}^{[t]}\neq\rho\underline{d}^{[t+n]}\textnormal{ for all }\rho\in\mathbb{F}_{q^{2n}}$.
            \medskip
            \\
            In (\ref{eq: H is also generated by the consecutive n-th powers}), it was observed that $\underline{d}^{[t+n+i]}\in\mathcal{H}$ for all $i\in\{1,\ldots,t-1\}$. Specifically, when $i=1$, there exist $\alpha_1,\ldots,\alpha_{t-1}\in\mathbb{F}_{q^{2n}}$ such that
            \[
                \underline{d}^{[t+n+1]}=\alpha_{1}\underline{d}^{[t+1]}+\ldots+\alpha_{t-2}\underline{d}^{[2t-2]}+\alpha_{t-1}\underline{d}^{[2t-1]},
            \] 
            from which, by raising to the power of $q$, we have
            \[
                \underline{d}^{[t+n+2]}=\alpha_{1}^{[1]}\underline{d}^{[t+2]}+\ldots+\alpha_{t-2}^{[1]}\underline{d}^{[2t-1]}+\alpha_{t-1}^{[1]}\underline{d}^{[2t]}.
            \]
            If $\alpha_{t-1}^{[1]}\neq 0$, then
            \[
                \underline{d}^{[2t]}=\frac{1}{\alpha_{t-1}^{[1]}}\left(\underline{d}^{[t+n+2]}-\alpha_{1}^{[1]}\underline{d}^{[t+2]}+\ldots+\alpha_{t-2}^{[1]}\underline{d}^{[2t-1]}\right)\in\mathcal{H},
            \]
            which is a contradiction since this would imply that
            \[
            \mathrm{rank}\begin{pNiceMatrix}
                        \underline{d}^{[t+1]}\\
                        \vdots\\
                        \underline{d}^{[2t-1]}\\
                        \underline{d}^{[2t]}
                    \end{pNiceMatrix}
            \]
            is $t-1$. This means that 
            \[
                \underline{d}^{[t+n+2]}=\alpha_{1}^{[1]}\underline{d}^{[t+2]}+\ldots+\alpha_{t-2}^{[1]}\underline{d}^{[2t-1]},
            \]
            and hence
            \[
                \underline{d}^{[t+n+3]}=\alpha_{1}^{[2]}\underline{d}^{[t+3]}+\ldots+\alpha_{t-3}^{[2]}\underline{d}^{[2t-1]}+\alpha_{t-2}^{[2]}\underline{d}^{[2t]}. 
            \]
            If $\alpha_{t-2}^{[2]}\neq 0$, 
            \[
            \underline{d}^{[2t]}=\frac{1}{\alpha_{t-2}^{[2]}}\left(\underline{d}^{[t+n+3]}-\alpha_{1}^{[2]}\underline{d}^{[t+3]}+\ldots+\alpha_{t-3}^{[2]}\underline{d}^{[2t-1]}\right)\in\mathcal{H},
            \]
            which is again a contradiction. On the other hand, if $\alpha_{t-2}^{[2]}=0$, then
            \[
                \underline{d}^{[t+n+3]}=\alpha_{1}^{[2]}\underline{d}^{[t+3]}+\ldots+\alpha_{t-3}^{[2]}\underline{d}^{[2t-1]}
            \]
            and
            \[
                \underline{d}^{[t+n+4]}=\alpha_{1}^{[3]}\underline{d}^{[t+4]}+\ldots+\alpha_{t-4}^{[3]}\underline{d}^{[2t-1]}+\alpha_{t-3}^{[3]}\underline{d}^{[2t]}
            \]
            and so on. In the $j$-th step, where $j\in\{2,\ldots,t-1\}$, we have the following situation 
            \[
                \underline{d}^{[t+n+j]}=\alpha_{1}^{[j-1]}\underline{d}^{[t+j]}+\ldots+\alpha_{t-j}^{[j-1]}\underline{d}^{[2t-1]}+\alpha_{t-(j-1)}^{[j-1]}\underline{d}^{[2t]}.
            \]
            If $\alpha_{t-(j-1)}^{[j-1]}=0$ in each $j$-th step for $j\in\{2,\ldots,t-2\}$, then at the $j=(t-1)$-th step, we have
            \[
                \underline{d}^{[2t+n-1]}=\alpha_{1}^{[t-2]}\underline{d}^{[2t-1]}+\alpha_{2}^{[t-2]}\underline{d}^{[2t]}.
            \]
            At this point, if $\alpha_{2}^{[t-2]}\neq 0$ then
            \[
            \underline{d}^{[2t]}=\frac{1}{\alpha_{2}^{[t-2]}}\left(\underline{d}^{[2t+n-1]}-\alpha_{1}^{[t-2]}\underline{d}^{[2t-1]}\right)\in\mathcal{H},
            \]
            which is a contradiction, since otherwise
            \[
                \underline{d}^{[2t+n-1]}=\alpha_{1}^{[t-2]}\underline{d}^{[2t-1]}.
            \]
            Note that $\alpha_{1}^{[t-2]}$ cannot be zero, because if it were, $\underline{d}$ would be the null vector. Therefore, by raising to the power of $q^{2n-t+1}$, it follows that
            \[
                \underline{d}^{[t]}=\frac{1}{\alpha_1^{[2n-1]}}\underline{d}^{[t+n]},
            \]
            which contradicts the hypothesis of the case in question.
            \end{proof}
            At this point Theorem \ref{corollary: The dimension of the solution space of S_exp is 1} can be proved.
            \begin{proof}
                We recall that $\mathrm{rank}(D_{\mathrm{exp}}\cdot A^{(t)})\leq\mathrm{rank}(A^{(t)})$ as $\ker (A^{(t)})\subseteq\ker (D_{\mathrm{exp}}\cdot A^{(t)})$. Consequently, in accordance with the rank-nullity theorem,
            \[
            t+1-\mathrm{rank}(A^{(t)})\leq t+1-\mathrm{rank}(D_{\mathrm{exp}}\cdot A^{(t)}),
            \]
            which proves our first claim. Using this result and the Sylvester's inequality, we get that
            \[
            t=\mathrm{rank}(A^{(t)})\geq\mathrm{rank}(D_{\mathrm{exp}}\cdot A^{(t)})\geq \mathrm{rank}(D_{\mathrm{exp}})+\mathrm{rank}(A^{(t)})-t=t.
            \]
            In conclusion, the dimension of the solution space of 
            \[
                S_{\mathrm{exp}}\cdot\underline{\Lambda}^{\top}=(D_{\mathrm{exp}}\cdot A^{(t)})\cdot\underline{\Lambda}^{\top}=\underline{0}^{\top}
            \]
            is one.
            \end{proof}
            The error span polynomial can then be determined as an immediate consequence. 
            \\
            \par
            Subsequently, the reconstruction of the sent codeword is achieved through the reiteration of the steps outlined in Subsections \ref{subsection: Finding the Root Space} and \ref{subsection: Determining the error}, i.e., once again, we determine a potential decomposition of the error $\underline{e}=\underline{a}\cdot B$ and reconstruct $\underline{c}=\underline{r}-\underline{e}$.
        \\
        \par
        The syndrome-based decoding algorithm of Trombetti-Zhou codes can now be summarized in Algorithm \ref{alg:Decoding of Trombetti-Zhou codes: a new syndrome-based decoding approach.}. Note that the decoding process is likely to fail if $t>\lfloor(d-1)/2\rfloor$.
        \subsection{Complexity Analysis}
        We analyze the complexity of Algorithm \ref{alg:Decoding of Trombetti-Zhou codes: a new syndrome-based decoding approach.} step by step.
        \begin{itemize}
            \item Solving the homogeneous linear system $S_{\mathrm{exp}}\cdot\underline{\Lambda}^{\top}=\underline{0}^{\top}$ (steps $6$-$10$) requires at most $\mathcal{O}(n^3)$ operations over $\mathbb{F}_{q^{2n}}$.
            \item The iterating process in steps $13$-$15$ can be repeated up to $\left\lfloor\frac{2n-(k+1)}{2}\right\rfloor$ times. Therefore, thanks to the $q$-circulant structure of the matrix $S^{(u)}$ and the considerations made regarding Lemma \ref{lemma: Rank of Syndrome Matrix}, steps $11$-$16$ also have complexity at most $\mathcal{O}(n^3)$ over $\mathbb{F}_{q^{2n}}$.
            \item Step $17$ requires the resolution of a linear system of equations of size $2n$ over $\mathbb{F}_q$, which has complexity $\mathcal{O}((2n)^3)$ over $\mathbb{F}_q$.
            \item Determining $(d_0,d_1,\ldots,d_{t-1})\in\mathbb{F}_{q^{2n}}$ (step $19$) takes $\mathcal{O}((2n)^2)$ operations over $\mathbb{F}_{q^{2n}}$ thanks to the recursive algorithm from \cite{gabidulin_TheoryOfCodesWithMaximumRankDistance} that exploit the $q$-Vandermonde structure of the coefficient matrix of the system involved. 
        \end{itemize}
        Since the complexity of the last steps is negligible, the overall complexity of the proposed decoding algorithm is cubic in half of the code length over $\mathbb{F}_{q^{2n}}$, i.e., it is in the order of $\mathcal{O}(n^3)$ operations over $\mathbb{F}_{q^{2n}}$.
            \begin{algorithm}[!ht]
            \DontPrintSemicolon
            \caption{Decoding of Trombetti-Zhou codes: a new syndrome-based decoding approach.}
            \label{alg:Decoding of Trombetti-Zhou codes: a new syndrome-based decoding approach.}
            \KwInput{$\underline{r}=(r_0,r_1,\dots,r_{2n-1})\in\mathbb{F}_{q^{2n}}^{2n}$}
            Set $H$ as in (\ref{eq: FqnGenerator and Fqn parity check matrix of TZcodes})\;
            $\mathbb{F}_{q^n}$\textit{-syndrome} calculation: $\underline{s}\gets \underline{r}\cdot H^{\top}\in\mathbb{F}_{q^{2n}}^{4n-2k}$ 
            \;
            \If{$\mathrm{Tr}_{q^{2n}/q^n}(\underline{s})=\underline{0}$}
                {
                    Estimated codeword: $\underline{c}\gets\underline{r}$
                }
            \Else
                {
                    \If{$k$ is even and $\mathrm{rank}(S_{\mathrm{exp}})=n-\frac{k}{2}$}
                    {   
                        Set up $t=n-\frac{k}{2}$\;
                        Solve $S_{\mathrm{exp}}\cdot\underline{\Lambda}^{\top}=\underline{0}^{\top}$ for $\underline{\Lambda}=(\Lambda_0,\Lambda_1,\dots,\Lambda_t)\in\mathbb{F}_{q^{2n}}^{t+1}$
                        \;
                        \If{$\left(\frac{\Lambda_0}{\Lambda_t},\frac{\Lambda_1}{\Lambda_t},\dots,\frac{\Lambda_{t-1}}{\Lambda_t},1\right)\notin\mathbb{F}_{q^{n}}^{t+1}$}
                            {
                                Declare decoding failure
                            }
                    }
                    \Else
                    {
                        Set up $S^{(t)}$ as in (\ref{eq:shape matrix in theorem rank of syndrome matrix}) for $t=\lfloor(2n-(k+1))/2\rfloor$ \;
                    \While{$\mathrm{rank}(S^{(t)})<t$}
                        {
       	                $t\gets t-1$\;
                            Set up $S^{(t)}$ as in (\ref{eq:shape matrix in theorem rank of syndrome matrix})
                        }
                    Solve $S^{(t)}\cdot\underline{\Lambda}^{\top}=\underline{0}$ for $\underline{\Lambda}=(\Lambda_0,\Lambda_1,\dots,\Lambda_t)\in\mathbb{F}_{q^{2n}}^{t+1}$ \
                    }
                    Find basis $(a_0,a_1,\dots,a_{\varepsilon-1})\in\mathbb{F}_{q^{2n}}^{\varepsilon}$ of the root space of $\Lambda(x)=\sum_{i=0}^{t}\Lambda_ix^{q^i}$ over $\mathbb{F}_{q^{2n}}$ \;
                    \If{$\varepsilon=t$}
                        {
                            Find $(d_0,d_1,\dots,d_{t-1})\in\mathbb{F}_{q^{2n}}^{t}$ by solving (\ref{eq:system to obtain d in TZ codes})\;
                            Find $B=(B_{i,j})_{j\in\{0,\dots,2n-1\}}^{i\in\{0,\dots,t-1\}}\in\mathbb{F}_q^{t\times 2n}$ such that $d_i=\sum_{j=0}^{2n-1}B_{i,j}\mu_j^{q^k}$ \;
                            Estimated codeword: $\underline{c}\gets \underline{r}-\underline{a}\cdot B$
                        }
                    \Else
                        {
    	                Declare decoding failure
                        }
                }
            \KwOutput{Estimated codeword $\underline{c}\in\mathbb{F}_{q^{2n}}^{2n}$ or ``decoding failure''}
        \end{algorithm}
        \section{Conclusions and open problems} 
        \label{sec: Conclusions and open problems}
        The first main contribution of this work is that we provide new insights into the theory of rank-metric codes that are linear over a subfield of the whole field of definition, along with new structural results for the family of Trombetti-Zhou codes seen as evaluation codes of length $2n$. In such a case we discussed the commutativity of Diagram \ref{tikz: diagram} which, unfortunately, does not hold when we evaluate $\mathcal{D}_{k}(\gamma)$ in $k\leq\ell<2n$  $\mathbb{F}_q$-linearly independent elements of $\mathbb{F}_{q^{2n}}$ (Lemma \ref{lemma: evaluate in less elements}).
        \par
        Our second main contribution lies all the theoretical foundations that ensure the correctness of the syndrome-based decoding algorithm for Trombetti-Zhou we designed in Algorithm \ref{alg:Decoding of Trombetti-Zhou codes: a new syndrome-based decoding approach.}. Future studies may overcome the hypothesis we had to impose on the error when its rank reaches the unique decoding radius by showing that, in truth, the proposed algorithm also works for correcting error vectors with entries in $\mathbb{F}_{q^{2n}}$ of rank weight $n-\frac{k}{2}$ when $k$ is even. This would also significantly simplify the algorithm as we would no longer need the \emph{if condition} at steps $9$-$10$ in Algorithm \ref{alg:Decoding of Trombetti-Zhou codes: a new syndrome-based decoding approach.}, which purely arises from the constraints of our current hypothesis. \par Although our algorithm is not as efficient as the one presented in \cite{On_interpolation-based_decoding_of_a_class_of_maximum_rank_distance_codes}, since it achieves a computational complexity of $\mathcal{O}(n^3)$ operations over $\mathbb{F}_{q^{2n}}$, it provides a more natural framework for designing a rank-metric decoder that can handle erasures and deviations, in analogy to the one of Silva, K\"{o}tter and Kschischang in \cite{RankMetric_approach_to_ErrorControl_in_RNetworkCoding} for Gabidulin codes.

\section*{Acknowledgments}

We would like to thank Violetta Weger for fruitful discussions and suggestions.
\par The second author thanks the Selmer Center in Secure Communication at the University of Bergen, Norway, for hosting her at the beginning of this project.  
\par The third and the last authors were partially supported by the Italian National Group for Algebraic and Geometric Structures and their Applications (GNSAGA - INdAM).

\section*{Appendix A}
\renewcommand{\thetheorem}{A.\arabic{theorem}}
\setcounter{theorem}{0}
In what follows, for the several technical proofs, we will use again the notation $[i]=q^i$.
\begin{proposition}
            \label{prop: anydecompositionissuitable}
            Any decomposition of the error vector $\underline{e}$ in (\ref{eq: rank decomposition of the error vector}) as product of a vector $\underline{a}$ of rank weight $t$ and a matrix $B\in\mathbb{F}_{q}^{t\times 2n}$ of rank $t$ is suitable for decoding.
        \end{proposition}
        \begin{proof}
            Let $\underline{a}=(a_0,a_1,\ldots,a_{t-1})$ and $\underline{a}^{\prime}=(a^{\prime}_0,a^{\prime}_1,\ldots,a^{\prime}_{t-1})$ be two $\mathbb{F}_q$-bases for the root space of $\Lambda(x)$ (see the previous paragraph \textit{Finding the Root Space of $\Lambda(x)$}), i.e.
            \[
                \left<a_0,a_1,\ldots,a_{t-1}\right>_{\mathbb{F}_q}=\left<a^{\prime}_0,a^{\prime}_1,\ldots,a^{\prime}_{t-1}\right>_{\mathbb{F}_q};
            \]
            then, there exists $M\in\mathrm{GL}_t(\mathbb{F}_q)$ such that $\underline{a}^{\prime}=\underline{a}\cdot M$. As observed, respectively, in the proof of Lemma \ref{lemma: Rank of Syndrome Matrix} and in Equation (\ref{eq: decomposition of S_exp}), 
            \[
                S^{(t)}=D^{(t)}\cdot A^{(t)}
            \]
            and
            \[
                S_{\mathrm{exp}}=D_{\mathrm{exp}}\cdot A^{(t)},
            \]
            where $A^{(t)}$ is known since its columns are $q$-powers of $\underline{a}$. If we concentrate on the first row of $S^{(t)}$ (that is the same of $S_{\mathrm{exp}}$) and raise each element to the power of $q^{2n-t-1}$, then
            \begin{align*}
                &\left(s_{2(t+1)-1}^{[2n-t-1]},s_{2t-1}^{[2n-t-2]},\ldots,s_1^{[2n-1]}\right)=\\ 
                &(d_0,d_1,\ldots,d_{t-1})\cdot\left(A^{(t)}\right)^{[2n-t-1]}=\\
                &(d_0,d_1,\ldots,d_{t-1})\cdot\begin{pmatrix}
                    \underline{a}^{{[2n-t-1]}^{\top}} & \underline{a}^{{[2n-t-2]}^{\top}} & \ldots &\underline{a}^{{[2n-1]}^{\top}}
                \end{pmatrix}.
            \end{align*}
            Therefore, by transposing both terms, 
            \[
            \begin{pmatrix}
                \underline{a}^{[2n-t-1]}\\
                \underline{a}^{[2n-t-2]}\\
                \vdots\\
                \underline{a}^{[2n-1]}\\
            \end{pmatrix}\cdot\begin{pmatrix}
                                    d_0\\
                                    d_1\\
                                    \vdots\\
                                    d_{t-1}\\
                                \end{pmatrix}=\begin{pmatrix}
                                    s_{2(t+1)-1}^{[2n-t-1]}\\
                                    s_{2t-1}^{[2n-t-2]}\\
                                    \vdots\\
                                    s_1^{[2n-1]}
                                \end{pmatrix}.
            \]
            Let
            \[
            \bar{A}\coloneqq\begin{pmatrix}
                \underline{a}^{[2n-t-1]}\\
                \underline{a}^{[2n-t-2]}\\
                \vdots\\
                \underline{a}^{[2n-1]}\\
            \end{pmatrix}
            \]
            and
            \[
            \underline{\check{s}}\coloneqq\begin{pmatrix}
                                    s_{2(t+1)-1}^{[2n-t-1]}\\
                                    s_{2t-1}^{[2n-t-2]}\\
                                    \vdots\\
                                    s_1^{[2n-1]}
                                \end{pmatrix},
            \]
            then
            \begin{equation}
                \label{eq: eq in any decomposition is good}
                \bar{A}\cdot\underline{d}^{\top}=\underline{\check{s}}.
            \end{equation}
            On the other hand, if we define 
            \[
            \bar{A}^{\prime}\coloneqq\begin{pmatrix}
                \underline{a}'^{[2n-t-1]}\\
                \underline{a}'^{[2n-t-2]}\\
                \vdots\\
                \underline{a}'^{[2n-1]}\\
            \end{pmatrix},
            \]
            since $M\in\mathrm{GL}_t(\mathbb{F}_q)$, then 
            \[
                \bar{A}^{\prime}=\bar{A}\cdot M.
            \]
            Therefore, by replacing the latter in (\ref{eq: eq in any decomposition is good}),
            \[
            \bar{A}^{\prime}\cdot M^{-1}\cdot\underline{d}^{\top}=\bar{A}\cdot\underline{d}^{\top}=\underline{\check{s}}
            \]
            and hence
            \[
                \underline{d}^{\prime^{\top}}=M^{-1}\cdot\underline{d}^{\top}.
            \]
            On the other hand, due to (\ref{eqn: vector of error locators}), $\underline{d}^{\top}=B\cdot\underline{\mu}^{{[k]}^{\top}}$, therefore
            \[
                B^{\prime}\cdot\underline{\mu}^{{[k]}^{\top}}=\underline{d}^{\prime^{\top}}=M^{-1}\cdot\underline{d}^{\top}=M^{-1}\cdot B\cdot\underline{\mu}^{{[k]}^{\top}},
            \]
            from which $B^{\prime}=M^{-1}\cdot B$. To conclude
            \[
            \underline{e}^{\prime}=\underline{a}^{\prime}\cdot B^{\prime}=\left(\underline{a}\cdot M\right)\cdot \left(M^{-1}\cdot B\right)=\underline{a}\cdot\left(M\cdot M^{-1}\right)\cdot B=\underline{a}\cdot B,
            \]
            that is what we wanted to prove.
        \end{proof}
        \begin{proof}[\textnormal{\textbf{Proof of Theorem \ref{thm: how we can rewrite the entries of the syndrome thanks to the vector of error locators}}}]
        We recall that the vector of error locators is $\underline{d}^{\top}=B\underline{\mu}^{{q^k}^{\top}}$ and that $\underline{d}^{{q^i}^\top}=B\underline{\mu}^{{q^{i+k}}^{\top}}$ as $B\in\mathbb{F}_{q}^{t\times 2n}$. By exploiting the bilinearity of the standard inner product, we get 
        \begin{align*}
                s_0&=a(\gamma\xi)^{[2n-k]}+\left<\underline{e},\gamma^{[2n-k]}\underline{\mu}\right> \\
                &=a(\gamma\xi)^{[2n-k]}+\gamma^{[2n-k]}\left<\underline{e},\underline{\mu}\right>\\
                &=a(\gamma\xi)^{[2n-k]}+\gamma^{[2n-k]}\left<\underline{a}B,\underline{\mu}\right>\\
                &=a(\gamma\xi)^{[2n-k]}+\gamma^{[2n-k]}\left(\underline{a}B\right)\cdot\underline{\mu}^{\top}\\
                &=a(\gamma\xi)^{[2n-k]}+\gamma^{[2n-k]}\underline{a}\cdot\left(B\underline{\mu}^{\top}\right)\\
                &=a(\gamma\xi)^{[2n-k]}+\gamma^{[2n-k]}\underline{a}\cdot\underline{d}^{{[2n-k]}^{\top}},
            \end{align*}
        for all $i\in\{1,\dots,2n-k-1\}$
            \[
            \begin{cases}
                    s_{2i-1}=\left<\underline{e},\underline{\mu}^{[k+i]}\right>=\left<\underline{a}B,\underline{\mu}^{[k+i]}\right>=\underline{a}\cdot\left(B\underline{\mu}^{{[k+i]}^{\top}}\right)=\underline{a}\cdot\underline{d}^{{[i]}^{\top}}\\
                    s_{2i}=\gamma s_{2i-1}
                \end{cases},
            \]
            and
            \begin{align*}
                s_{4n-2k-1}&=b\gamma\xi+\left<\underline{e},\underline{\mu}^{[k]}\right>\\
                &=b\gamma\xi+\left(\underline{a}B\right)\cdot\underline{\mu}^{{[k]}^{\top}}\\
                &=b\gamma\xi+\underline{a}\cdot \left(B\underline{\mu}^{{[k]}^{\top}}\right)\\
                &=b\gamma\xi+\underline{a}\cdot\underline{d}^{\top}.
            \end{align*}
            \end{proof}
             \begin{proof}[\textnormal{\textbf{Proof of Theorem \ref{thm: obtain more information thanks to the trace}}}]
            Since both $a,b\in\mathbb{F}_{q^n}$, $\underline{a}\in\mathbb{F}_{q^n}^t$ and $\mathrm{Tr}_{q^{2n}/q^n}(\gamma\xi)=\mathrm{Tr}_{q^{2n}/q^n}\left((\gamma\xi)^{[2t]}\right)=0$, we have
            \begin{align*}
                \tilde{s}_0&=\mathrm{Tr}_{q^{2n}/q^n}(s_0)\\
                &=a\mathrm{Tr}_{q^{2n}/q^n}\left((\gamma\xi)^{[2t]}\right)+\mathrm{Tr}_{q^{2n}/q^n}\left(\gamma^{[2t]}\underline{a}\cdot\underline{d}^{{[2t]}^{\top}}\right)\\
                &=\mathrm{Tr}_{q^{2n}/q^n}\left(\gamma^{[2t]}\underline{a}\cdot\underline{d}^{{[2t]}^{\top}}\right)\\
                &=\mathrm{Tr}_{q^{2n}/q^n}\left(\gamma^{[2t]}\displaystyle\sum_{\ell=0}^{t-1}a_{\ell}d_{\ell}^{[2t]}\right)\\
                &=\displaystyle\sum_{\ell=0}^{t-1}a_{\ell}\mathrm{Tr}_{q^{2n}/q^n}\left(\gamma^{[2t]}d_{\ell}^{[2t]}\right)\\
                &=\underline{a}\cdot \mathrm{Tr}_{q^{2n}/q^n}\left(\gamma^{[2t]}\underline{d}^{[2t]}\right)^{\top},
            \end{align*}
            for all $i\in\{1,\dots,2t-1\}$
            \[
            \begin{cases}
                    \tilde{s}_{2i-1}=\mathrm{Tr}_{q^{2n}/q^n}\left(s_{2i-1}\right)=\mathrm{Tr}_{q^{2n}/q^n}\left(\underline{a}\cdot\underline{d}^{{[i]}^{\top}}\right)=\underline{a}\cdot \mathrm{Tr}_{q^{2n}/q^n}\left(\underline{d}^{[i]}\right)^{\top}\\
                    \tilde{s}_{2i}=\mathrm{Tr}_{q^{2n}/q^n}(\gamma s_{2i-1})=\underline{a}\cdot \mathrm{Tr}_{q^{2n}/q^n}\left(\gamma\underline{d}^{[i]}\right)^{\top}
                \end{cases},
            \]
            and
            \[
            \tilde{s}_{4t-1}=\mathrm{Tr}_{q^{2n}/q^n}(s_{4t-1})=\underline{a}\cdot \mathrm{Tr}_{q^{2n}/q^n}\left(\underline{d}\right)^{\top}.
            \]
        \end{proof}
\bibliographystyle{abbrv}
\bibliography{bibliography}
\bigskip

\noindent Chunlei Li,\\
Department of Informatics,\\ 
University of Bergen, Norway.\\
E-mail: chunlei.li@uib.no\\

\medskip
\noindent Angelica Piccirillo,\\
Department of Mathematics,\\ 
Technical University of Munich,\\
TUM School of Computation, Information
and Technology (CIT), Germany.\\ 
E-mail: angelica.piccirillo@tum.de\\

\medskip
\noindent Olga Polverino and Ferdinando Zullo,\\
Dipartimento di Matematica e Fisica,\\ 
Universit\`a degli Studi della Campania ``Luigi Vanvitelli'', Italy.\\
E-mail: \{olga.polverino, ferdinando.zullo\}@unicampania.it
\end{document}